\documentclass[11pt,letterpaper]{article}
\pdfoutput=1

\usepackage[T1]{fontenc}
\usepackage{lmodern}
\usepackage[margin=1in]{geometry}
\usepackage{microtype}
\usepackage{amsmath,amssymb,amsthm,mathtools}
\usepackage{graphicx}
\usepackage{booktabs,array,tabularx}
\usepackage{enumitem}
\usepackage[numbers,sort&compress]{natbib}
\usepackage{xcolor}
\definecolor{linkblue}{rgb}{0.12,0.28,0.68}
\usepackage[linktocpage=true]{hyperref}
\hypersetup{
  colorlinks=true,
  citecolor=linkblue,
  linkcolor=linkblue,
  urlcolor=linkblue,
  pdfauthor={Ning Bao, Christian Ferko},
  pdftitle={Beyond Strong Subadditivity in Holographic RG Flows}
}

\usepackage{fancyhdr}

\fancypagestyle{firstpage}{%
  \fancyhf{}
  \fancyhead[R]{\today}
  \fancyfoot[C]{\thepage}
  
}

\newtheorem{assumption}{Assumption}[section]

\newtheorem{theorem}[assumption]{Theorem}

\newtheorem{corollary}[assumption]{Corollary}
\newtheorem{proposition}[assumption]{Proposition}
\theoremstyle{remark}

\newcommand{\dd}{\mathrm d}
\newcommand{\QFive}{Q_5}
\newcommand{\PhiFive}{\Phi_5}

\numberwithin{equation}{section}
\allowdisplaybreaks
\title{\bfseries Beyond Strong Subadditivity: Holographic Entropy Inequalities Along Renormalization Group Flows}
\author{
  Ning Bao$^{a,b}$ and Christian Ferko$^{a,c}$\\[0.8em]
  \small $^a$Department of Physics, Northeastern University, Boston, MA 02115, USA\\
  \small $^b$Computational Science Initiative, Brookhaven National Laboratory,
  Upton, NY 11973, USA\\
  \small $^c$The NSF Institute for Artificial Intelligence and Fundamental Interactions\\[0.5em]
  \small
  \texttt{ningbao75@gmail.com}
  \qquad
  \texttt{c.ferko@northeastern.edu}
}
\date{}

\begin{document}

\maketitle
\thispagestyle{firstpage}

\begin{abstract}
Strong subadditivity (SSA) on a common light cone gives the Casini-Huerta entropic proof of the three-dimensional $F$-theorem.  We ask whether holographic entropy inequalities beyond SSA similarly constrain renormalization group flows for which every intermediate theory admits a semiclassical holographic description.  For small, disjoint deformations of a common light-cone region, we show that the second-order response of a broad class of balanced holographic inequalities depends only on pairwise correlations already controlled by SSA. A six-party example shows that the full finite inequality nevertheless contains genuinely multipartite information, so its disappearance is a limitation of the second-order expansion rather than of the inequality itself. Two natural finite constructions do not recover the missing information. We nevertheless find two ways in which information beyond SSA survives. A five-party inequality bounds the rate at which a conditional correlation grows as one region is enlarged. Separately, a continuum limit of the odd-cyclic inequalities gives a constraint on the angular shape dependence of entanglement entropy, and Lorentz symmetry relates this constraint to radial evolution. Thus holographic entropy inequalities beyond SSA do constrain entanglement along RG flows, although we do not obtain a second universal analogue of the $F$-function.
\end{abstract}

\clearpage
\setcounter{tocdepth}{2}
\tableofcontents
\clearpage

\section{Introduction}
\label{sec:introduction}

Renormalization group (RG) flow imposes an ordering on quantum field
theories: short-distance degrees of freedom may disappear at long distances,
but they cannot be recovered by reversing the coarse graining.  One of the
most useful ways to make this statement quantitative is to construct a
function of scale that decreases monotonically and reduces, at conformal
fixed points, to a universal measure of degrees of freedom.  In two and three
spacetime dimensions, entanglement entropy provides such functions without
requiring a choice of running couplings
\cite{Zamolodchikov:1986gt,Casini:2004bw,Casini:2012ei}.

The three-dimensional example will be our starting point.  Let $S(R)$ be the
vacuum entanglement entropy of a disk of radius $R$.  Strong subadditivity
(SSA), applied to boosted disks whose boundaries approach a common light
cone, implies $S''(R)\leq0$.  Therefore the renormalized disk entropy
\begin{equation}
 \mathcal F(R):=R S'(R)-S(R),
 \label{eq:intro-F-function}
\end{equation}
obeys $\mathcal F'(R)=R S''(R)\leq0$.  At a conformal fixed point the local
perimeter term cancels from \eqref{eq:intro-F-function}, leaving the sphere
free energy $F$.  This is the Casini-Huerta $F$-theorem
\cite{Jafferis:2011zi,Casini:2012ei,Casini:2015woa}.  Its proof is especially
suggestive because the required inequality is kinematic, while the RG
information enters through a one-parameter family of regions.

At leading semiclassical order, holographic entanglement entropies are
computed by the Ryu-Takayanagi (RT) prescription in static bulk geometries
and by its covariant Hubeny-Rangamani-Takayanagi (HRT) extension
\cite{Ryu:2006bv,Hubeny:2007xt}.  Here and below, a \emph{party} is a
labeled boundary subsystem.  Holographic entropies obey monogamy of mutual
information (MMI),
$I(A:BC)\geq I(A:B)+I(A:C)$, where
$I(A:B):=S(A)+S(B)-S(AB)$, as well as an infinite family of multipartite
inequalities beyond SSA
\cite{Hayden:2011ag,Bao:2015bfa,HernandezCuenca:2019wgh}.  This raises a
natural question: if every theory along an RG flow has a semiclassical
holographic dual, can these additional inequalities be converted into new
constraints on the flow?  The most ambitious answer would be a new radial
monotone, independent of $\mathcal F(R)$.  A weaker but still useful answer is
an inequality that constrains scale derivatives or shape responses at each
point of the flow.

Several recent developments sharpen this question.  Entropic proofs have
been used to establish irreversibility directly in holographic RG geometries,\footnote{Here we note that when we say holographic RG, we mean renormalization group flows that are holographic at every point in the RG flow, as opposed to holographic renormalization in the bulk direction. We will mean the former in all instances in this paper when discussing holographic RG.}
while weighted-min-cut models of the cone of holographic entropy vectors (the
lists of entropies for all unions of a fixed set of parties), majorization
tests which compare the ordered partial sums of the effective region sizes
appearing on the two sides of a common-light-cone inequality, and interpretations of
entanglement-wedge depth as an RG scale have clarified how holographic
inequalities encode scale
\cite{Deddo:2024hll,Grimaldi:2025hec,Czech:2026maj,
Grimaldi:2026com,Czech:2026rgp}.  Related work has also emphasized that
multipartite holographic constraints need not be preserved by operations
which are harmless for ordinary entropy inequalities
\cite{HernandezCuencaHubenyJia2024Multipartite}.  Our aim is complementary:
we ask which of these inequalities survive controlled light-cone limits as
constraints on derivatives along a flow.

The analogy with the $F$-theorem is not automatic.  Every entropy term must
be built from one common set of disjoint boundary \emph{atoms}, by which we
mean fixed elementary regions whose unions furnish every argument of $S$.
We call such a simultaneous representation a common atomization; a
collection of individually valid null cuts need not admit one.  Moreover, a
generic holographic entropy inequality has a nonzero \emph{gap}, meaning the
value of its nonnegative linear entropy combination, before two nearby scales
are compared.  Dividing this gap by the shell thickness does not give a
finite differential statement unless the zeroth-order contribution cancels.
Finally, covariant HRT entropy inequalities can require geometric input
beyond the existence of the boundary regions.  These issues force us to
distinguish coefficient identities, fixed-scale inequalities, and statements
along one physical RG trajectory.

In this work, we organize the analysis around three questions.

\begin{enumerate}[label=(\roman*),leftmargin=2.1em]
 \item \emph{What obstructs the direct replacement of SSA?}
 We first characterize which families of light-cone cuts, with radial profiles
 additive in the party deformations, can be constructed from one common set of
 disjoint boundary atoms.  For the resulting disjoint-support bump family, we
 prove that every nonzero inequality which is centered (one distinguished
 region appears in every term), balanced (the local contribution of each atom
 cancels), and valid for all weighted-graph minimum cuts has a nontrivial pair
 projection.  In other words, some two-party cumulative coefficient is
 nonzero.  Its common-amplitude quadratic response is a nonnegative sum of
 pair susceptibilities, defined below as the negatives of mixed second shape
 derivatives; their signs are already controlled by SSA.  The
 six-party functional $G_6$ nevertheless contains information not generated
 by nonnegative sums of the elemental Shannon inequalities.  Exact removal
 of selected parties at the level of coefficient identities, finite
 combinations evaluated at several scales with one atomization held fixed,
 and a regulated domain-wall calculation sharpen this obstruction in
 complementary ways.

 \item \emph{Can a finite-party inequality nevertheless constrain radial
 growth?}
 We evade the fixed-atom obstruction by absorbing an already-grown petal
 $A_R$ into the conditioning system of the five-party inequality; here by
 ``petal'' we mean a boundary region occupying one angular slot on the
 regulated light cone.  Its remaining zeroth-order term is the conditional
 mutual information
 $I(C:D\mid ZA_R)$, where
 $I(X:Y\mid W):=S(XW)+S(YW)-S(W)-S(XYW)$.  On the corresponding Markov face -- the locus on which this quantity vanishes -- the inequality gives an upper bound on the right logarithmic rate at which a conditional mutual
 information grows as the petal is enlarged.  If the equality persists along
 an interval, the same bound makes a particular combination of conditional
 mutual informations monotone.  The upper bound is not a conic consequence of
 SSA, MMI, or their purifier images; that is, it cannot be obtained as a
 nonnegative linear combination of those inequalities.  A purifier image is obtained by
 exchanging a named subsystem with the complementary purifying system.  The
 bound is sharply saturated by a refining family of
 holographic graph models.  Because the spectators remain fixed, however,
 this is a mixed scale/shape response rather than a self-similar dilation of
 the complete configuration.

 \item \emph{Can one avoid the Markov hypothesis?}
 The odd-cyclic inequalities, defined for an odd number of cyclically ordered
 regions and built from entropies of consecutive blocks, provide a complementary
 answer.  Their infinite-party angular limit gives a beyond-SSA shape
 inequality without imposing an entropy equality.  A Lorentz Ward identity
 rewrites this result as a mixed radial/shape constraint.  The residual shape
 response has no definite sign, so the construction does not by itself
 produce a scalar radial monotone.
\end{enumerate}

The answer to our original question is therefore qualified.  Holographic
inequalities beyond SSA do constrain entanglement response along the flow,
but the two successful constructions retain different limitations.  The
$Q_5$ shell theorem uses a trajectory-dependent Markov cancellation and an
atomization that changes with the scale, whereas the odd-cyclic theorem takes
an infinite-party angular limit.  Neither construction conflicts with the
finite fixed-atom no-go results.

We work in three boundary spacetime dimensions and at leading classical
holographic order.  The finite graph inequalities are unconditional within
their stated domain.  Applications to covariant HRT entropies use the
common-slice hypotheses stated below or, for the $Q_5$ and odd-cyclic
results, assume covariant validity on the common regulated slice.  Null,
continuum, and endpoint limits require the differentiability and uniformity
assumptions given where they are used.  In particular, we do not claim a
universal beyond-SSA radial monotone for smooth single-boundary RG-flow
vacua.

The remainder of this article is organized as follows.  Section~\ref{sec:ssa-baseline} reviews the SSA mechanism behind the $F$-theorem, formulates the common-atom condition for light-cone cuts, and derives the pair projection of a general centered inequality.  Section~\ref{sec:obstructions} tests direct finite-atom extensions using the six-party inequality, exact deletion, finite scale differences, and an independent bulk calculation.  In
Section~\ref{sec:beyond-ssa}, we relax the fixed-atom assumptions and derive
the five-party shell bound and the odd-cyclic continuum constraint.
Section~\ref{sec:conclusion} summarizes our results and presents directions for future research. Ancillary material which provides further details about certain results in the body of this work has been collected in Appendices \ref{app:selectors} and \ref{app:bulk-pair-response}.

\section{Strong subadditivity and pair rigidity on the light cone}
\label{sec:ssa-baseline}

The Casini-Huerta proof will serve both as our benchmark and as a diagnostic
for proposed generalizations.  We begin by recalling how SSA, the Markov
property of the UV vacuum, and a controlled null limit combine to produce the
$F$-function.  We then formulate the common-atom condition for multipartite
light-cone cuts and study the simplest deformations which admit such a
 realization.  This analysis makes the first obstruction transparent: at
 quadratic order, every centered, balanced inequality valid on graph entropy
 vectors in the family considered below is projected onto pair
 susceptibilities whose signs are already fixed by SSA.

Consider a Lorentz-invariant vacuum in $2+1$ dimensions and let $S(R)$ denote
the entropy of a round disk $D_R$ on a fixed Cauchy slice.  The regulated
entropy has a local perimeter divergence, but the combination
$\mathcal F(R)$ in \eqref{eq:intro-F-function} is finite.  Casini and Huerta
apply SSA to many boosted disks whose boundaries lie on a common null cone.
In the continuum limit, the union and intersection approach round disks with
nearby radii, while the local contributions supported at their
short-wavelength junctions cancel.  The result is the concavity inequality
\begin{equation}
 S''(R)\leq0.
 \label{eq:disk-concavity}
\end{equation}
It follows immediately that
\begin{equation}
 \mathcal F'(R)=R S''(R)\leq0.
 \label{eq:F-monotonicity}
\end{equation}
At a conformal fixed point,
\begin{equation}
 S_*(R)=\alpha_\delta\frac{R}{\delta}-F_*,
 \label{eq:disk-CFT-entropy}
\end{equation}
so that $\mathcal F(R)=F_*$.  A flow from a UV CFT to an IR CFT therefore
obeys $F_{\rm UV}\geq F_{\rm IR}$.

There are three ingredients in this argument.  First, all regions in SSA are
unions of the same boundary atoms.  Second, in the regulator-controlled
null-cone formulation one subtracts the UV CFT contribution.  The UV CFT
vacuum is a quantum Markov state on the light cone: it saturates SSA, so its
SSA gap vanishes and removes the regulator-sensitive UV terms
\cite{Casini:2017vbe,Casini:2017roe,Casini:2018kzx}.  Third,
the remaining finite difference has a controlled local limit.  The first
ingredient is geometric, the second removes the regulator-sensitive
zeroth-order contribution, and the third turns an entropy inequality into a
derivative constraint.  We will use precisely this three-part test for
inequalities beyond SSA.

For later use, let $0$ denote a distinguished central region, write
$[n]:=\{1,\ldots,n\}$, and let $T\subseteq[n]$ specify a union of the
$n$ disjoint outer atoms.  We write $0T$ for the union of $0$ with the atoms
labeled by $T$.  A functional is \emph{centered} when $0$ appears in every
entropy term, so it has the form
\begin{equation}
 Q_c[S]=\sum_{T\subseteq[n]}c_T S(0T).
 \label{eq:centered-functional-setup}
\end{equation}
It is \emph{balanced} when
\begin{equation}
 \sum_Tc_T=0,
 \qquad
 \sum_{i \in T}c_T=0
 \quad (i=1,\ldots,n).
 \label{eq:balanced-setup}
\end{equation}
Balance cancels the leading local contribution associated with every atom.
We will also use weighted holographic graph models.  Each boundary party is a
distinguished terminal vertex, and $S(T)$ is the minimum total edge weight of
a cut separating the terminals in $T$ from the complementary terminals.  The
list of entropies for all unions is a graph entropy vector; a further terminal,
called the purifier, represents the complement of all named parties.  Such
vectors model the leading classical static RT entropy cone.  Holographic
entropy inequalities supply many coefficient arrays for which $Q_c[S]\geq0$
on graph or RT entropy vectors
\cite{Bao:2015bfa,HernandezCuenca:2019wgh,Daguerre:2022aon}.  The question is
whether their finite gaps can be converted into information about variation with respect to the physical scale $R$.

We will distinguish three deformation parameters throughout the paper.  The
radius $R$ is a radial size parameter and becomes the RG scale for a
self-similar family.  When only one petal grows while spectator regions remain
fixed, $R\partial_R$ instead describes nested-region growth, or equivalently a
mixed scale/shape direction after all positions are made dimensionless.  The
parameter $\epsilon$ controls outward boundary bumps at fixed $R$, while
$\kappa$ controls the transverse special conformal transformation in the
bulk calculation.  Positivity of a derivative with respect to one parameter
does not imply positivity with respect to either of the others.

The identification of $R$ with an RG scale is meant in
the following boundary-field-theory sense.  For a self-similar family
$\gamma^{(R)}(\phi)=R\,\widehat\gamma(\phi)$, varying $R$ is a global
dilation which leaves all dimensionless shape data fixed.  If the UV
theory is perturbed by
$\sum_\alpha\lambda_\alpha\int \dd^3x\,\mathcal O_\alpha$, then, after
the rescaling $x^\mu=R\widetilde x^\mu$, the path integral for the
corresponding reduced density matrix depends near the UV fixed point on
the dimensionless combinations
$\lambda_\alpha R^{3-\Delta_\alpha}$.  More generally, the
Callan--Symanzik equation expresses the $R$ dependence of a
regulator-independent balanced entropy combination in terms of the
running couplings at the inverse length scale $\mu\simeq R^{-1}$.
Thus increasing $R$ probes progressively lower energies, and $R\to0$
and $R\to\infty$ probe the UV and IR fixed points when the corresponding
limits exist \cite{Casini:2012ei,Casini:2015woa}.  In this sense $R$ is
a physical RG \emph{length} scale; it is neither the bulk holographic
radial coordinate nor a parameter which literally changes the
microscopic theory.  This interpretation is pure only for self-similar
dilations.  If other physical lengths are held fixed while $R$ varies,
their ratios to $R$ also change, and $R\partial_R$ instead generates a
mixed scale/shape response.

The following assumptions specify the physical setting.  They are stronger
than the corresponding statements for graph entropy vectors and will be invoked only where
needed.

\begin{assumption}[Holographic flow and regulated null limit]
Every theory along the flow admits a leading-order semiclassical holographic
description.  All entropy terms are first formed from fixed disjoint atoms on
one regulated spacelike Cauchy slice, with a common UV prescription, before
the null limit is taken.  The balanced gaps used below have finite,
regulator-independent null limits.
\end{assumption}

\begin{assumption}[Differential, continuum, and endpoint limits]
Whenever a shell inequality is divided by its logarithmic thickness, the
indicated one-sided rates exist.  Whenever the number of angular parties is
sent to infinity, the regulated gaps and symmetric shape derivatives converge
uniformly on the entropy branch under consideration, meaning a smooth phase
in which the relevant extremal-surface saddles remain fixed.  Endpoint statements
assume sufficient convergence to exchange the null limit with the UV or IR
CFT limit.  At a phase transition, the finite inequalities remain primary
and all derivative statements are understood one-sidedly.
\end{assumption}

The first task is now clear: before differentiating an entropy inequality, we
must determine whether its proposed null cuts are unions of one common set of
regions.  We turn to this geometric consistency problem next.

\subsection{Realizable light-cone cuts}
\label{sec:legality}
\label{subsec:legal-bumps}

The first requirement is geometric.  Every term in an entropy inequality
must be a union of the same fixed, disjoint boundary regions, which is
stronger than asking whether each proposed cut is individually well defined.
We therefore begin by characterizing simultaneous realizability and then
construct a deformation family that satisfies this condition.

The problem is one-dimensional on each null generator.  Parameterize the
cone by $\phi\in S^1$ and an affine radial coordinate $\varrho\geq0$.  A
positive profile $\gamma(\phi)$ specifies the cut
$\varrho=\gamma(\phi)$ and defines the region containing the tip,
\begin{equation}
 A_\gamma=\{(\phi,\varrho):0\leq\varrho\leq\gamma(\phi)\}.
 \label{eq:cone-cut}
\end{equation}
Thus $A_\gamma\subseteq A_{\widetilde\gamma}$ if and only if
$\gamma\leq\widetilde\gamma$ almost everywhere.  Set identities below are
understood up to sets of angular measure zero.

Fix a base profile $\gamma_0$ and assign a nonnegative measurable thickness
$f_i(\phi)$ to each label $i$.  For $T\subseteq[n]$, consider the target
\begin{equation}
 \gamma_T=\gamma_0+\sum_{i\in T}f_i .
 \label{eq:additive-target}
\end{equation}

Let $\mathcal T$ be the finite collection of label sets $T$ whose target cuts
must be realized simultaneously.  The targets are
\emph{common-atom realizable}, or simply realizable, if there are fixed,
mutually disjoint regions $A_i$, also disjoint from
$A_0=A_{\gamma_0}$, such that
\begin{equation}
 A_0\cup\bigcup_{i\in T}A_i=A_{\gamma_T}
 \qquad (T\in\mathcal T),
 \label{eq:legal-common-atoms}
\end{equation}
where $A_i$ has thickness $f_i$ on each generator.  The same $A_i$ must
occur in every term containing label $i$.

At fixed $\phi$, labels with zero thickness are irrelevant.  Define
\begin{equation}
 L(\phi)=\{i:f_i(\phi)>0\},
 \qquad
 \mathcal T_\phi=\{T\cap L(\phi):T\in\mathcal T\}.
 \label{eq:active-projection}
\end{equation}
$\mathcal T_\phi$ records the active labels requested by the terms.  One
radial ordering realizes all targets precisely when these subsets form a
chain, meaning that every pair is comparable by inclusion.

\begin{theorem}
\label{thm:fiberwise-chain}
The targets \eqref{eq:additive-target} are common-atom realizable if and only
if $\mathcal T_\phi$ is totally ordered by inclusion for almost every
$\phi$.
\end{theorem}

\begin{proof}
Suppose first that fixed disjoint regions exist, and fix $\phi$ outside the
measure-zero exceptional set.  If two projected terms $C,D$ were
incomparable, there would be labels
 $i\in C\setminus D$ and $j\in D\setminus C$.  The region for $C$ is an
 initial radial interval extending from the cone tip and containing $A_i$ but
 not $A_j$, so $A_i$ must lie below $A_j$.  The
region for $D$ requires the opposite order.  This is impossible for two
intervals of positive length.

Conversely, suppose that the distinct projected terms form a chain
$C_0\subsetneq\cdots\subsetneq C_m$.  We construct the atoms directly on
the generator.  Starting at $\gamma_0(\phi)$, place the labels in $C_0$,
then those in $C_1\setminus C_0$, and so on, assigning label $i$ an interval
of length $f_i(\phi)$.  Labels absent from every requested term go last.
Each requested set is then an initial segment ending at
$\gamma_0+\sum_{i\in T}f_i$.  A fixed rule resolves ties, and finiteness of
the term list makes the construction measurable in $\phi$.
\end{proof}
The theorem requires no smoothness. Smoothness enters only in the shape derivatives below.

For a finite term list, define the \emph{separation graph}
$G_{\mathcal T}$ by joining $i$ and $j$ whenever two terms
$T,U\in\mathcal T$ obey
\begin{equation}
 i\in T\setminus U,
 \qquad j\in U\setminus T.
 \label{eq:separation-edge}
\end{equation}

An edge records incompatible radial orders.  The theorem is therefore
equivalent to requiring $L(\phi)$ to be an independent set of
$G_{\mathcal T}$, meaning that no two vertices in $L(\phi)$ are joined by an edge, for almost
every $\phi$.  For instance, intervals ordered as $x$ followed by $y$ realize
$\{\varnothing,\{x\},\{x,y\}\}$ but not $\{y\}$.

The criterion can be applied before evaluating any entropy.  To see this,
consider
\begin{align}
 Q_5={}&S(0ab)+S(0ac)+S(0ad)+S(0bc)+S(0bd) \notag\\
 &-S(0a)-S(0b)-S(0cd)-S(0abc)-S(0abd).
 \label{eq:Q5-legality}
\end{align}
Its term list separates every pair among $a,b,c,d$, so
$G_{\mathcal T}=K_4$, the complete graph on four labels, and
\begin{equation}
 f_i(\phi)f_j(\phi)=0
 \quad\text{a.e. for every }i\ne j.
 \label{eq:Q5-disjoint-support}
\end{equation}
The Fourier ansatz
\begin{equation}
 f_a=-q\cos2\phi,\qquad f_b=u,\qquad
 f_c=q\cos2\phi,\qquad f_d=u,
 \label{eq:illegal-one-mode}
\end{equation}
with $u\geq0$ and $q\in\mathbb R$ fails this condition.  For $q\ne0$,
$f_a$ and $f_c$ change sign, and the profiles $R-q\cos2\phi$ and $R$
exchange order.  For $q=0$, $f_b$ and $f_d$ overlap unless $u=0$.
Positive offsets $p_i+q_i\cos(2\phi-\alpha_i)$ do not help: for
$p_i\geq|q_i|$ they are positive almost everywhere, whereas
\eqref{eq:Q5-disjoint-support} permits at most one nonzero thickness on each
generator.  Thus this family is realizable only at $u=q=0$.

A useful sufficient construction assigns different labels to disjoint sets
of generators.  Let $g_i\geq0$ be smooth functions with pairwise disjoint
closed angular supports and $\epsilon_i\geq0$.  Define
\begin{equation}
 \gamma_T^{(R)}(\phi)
 =R\left(1+\sum_{i\in T}\epsilon_i g_i(\phi)\right).
 \label{eq:legal-bump-family}
\end{equation}
Here $R$ is the overall size and $\epsilon_i$ are outward amplitudes at fixed
$R$.  The base region ends at $\varrho=R$, while label $i$ occupies the
radial interval $R<\varrho\leq R(1+\epsilon_i g_i(\phi))$.  At most one
such interval is present on each generator, so
\eqref{eq:legal-bump-family} is realizable for any term list.  This is the
family used below.  The same family appears in the common-light-cone
majorization tests of
\cite{Grimaldi:2025hec,Grimaldi:2026com,Czech:2026maj}.

The cuts lie on a null cone, whereas entropy is evaluated on a spacelike
Cauchy slice.  We regulate the cone by
\begin{equation}
 t=f_{\eta,R}(\varrho)
 :=\sqrt{\varrho^2+(\eta R)^2}-\eta R,
 \qquad \eta>0.
 \label{eq:spacelike-cone-regulator-short}
\end{equation}
This regulator profile obeys $|f'_{\eta,R}|<1$.  After it is flattened beyond all cuts,
it defines a complete boundary Cauchy slice $\Sigma_{\eta,R}$.  The
regulated region is
\begin{equation}
 A_\gamma^{(\eta,R)}
 =\bigl\{(f_{\eta,R}(\varrho),\varrho,\phi):
 0\leq\varrho\leq\gamma(\phi)\bigr\}
 \subset\Sigma_{\eta,R}.
 \label{eq:spacelike-cut-approximant-short}
\end{equation}
All unions and intersections are taken at fixed $\eta$.  For a balanced
functional $Q$, define
\begin{equation}
 Q^{\rm null}(R):=\lim_{\eta\downarrow0}Q^{(\eta)}(R),
 \label{eq:controlled-null-gap-short}
\end{equation}
assuming that the limit exists and is regulator independent.  This is part
of the null-limit prescription.  It puts all boundary regions on one Cauchy
slice before the limit, but it does not make their HRT surfaces minimize on
one bulk slice.  The additional common-slice condition used for $G_6$ is
stated in Section~\ref{sec:finite-g6}.

\subsection{Pair rigidity and local response}
\label{sec:pair-rigidity}

Having constructed a common-atom-realizable deformation family, we now ask
what a general centered inequality can detect near the undeformed cut.  Here
\emph{pair rigidity} means that a nonzero balanced functional valid on graph
entropy vectors cannot have all of its two-party cumulative coefficients vanish.  The answer
is most transparent in terms of cumulative coefficients, which collect the
coefficients of all entropy terms containing a prescribed set of outer
parties.  Let $0$ denote the central region and
let $T\subseteq[n]$ label a union of the remaining regions.  A centered
functional is
\begin{equation}
 Q_c[S]=\sum_{T\subseteq[n]}c_T S(0T),
 \qquad Q_c[S]\geq0.
 \label{eq:centered-Q}
\end{equation}
In this subsection, \eqref{eq:centered-Q} is required to be nonnegative only
on classical holographic graph entropy vectors.  We also impose balance,
\begin{equation}
 \sum_Tc_T=0,
 \qquad
 \sum_{T\ni i}c_T=0\quad(i=1,\ldots,n).
 \label{eq:centered-balance-short}
\end{equation}
For $U\subseteq[n]$, define the cumulative coefficient
\begin{equation}
 m_U=\sum_{T\supseteq U}c_T.
 \label{eq:upper-incidence-short}
\end{equation}
Balance gives $m_\varnothing=m_i=0$.  We refer to $m_{ij}$ as the pair
coefficient and to the collection $\{m_{ij}\}_{i<j}$ as the pair projection
of $Q_c$.

For a balanced graph inequality, region dominance states that the sum of
positive coefficients of terms containing any fixed set $U$ cannot exceed the
absolute value of the sum of their negative coefficients.  Since $m_U$ is
their signed difference, this
implies $m_U\leq0$ for every $|U|\geq2$
\cite{Grimaldi:2026com}.  This does not by itself exclude all $m_{ij}=0$
with some $m_U\ne0$ for $|U|\geq3$.  The following theorem does.

\begin{theorem}
\label{thm:strict-pair}
Suppose \eqref{eq:centered-Q} is nonnegative on every classical holographic
graph entropy vector and obeys \eqref{eq:centered-balance-short}.  Then
\begin{equation}
 m_{ij}\leq0\qquad(i\ne j).
 \label{eq:pair-sign-short}
\end{equation}
If $c\ne0$, then $m_{ij}<0$ for some $i\ne j$.
\end{theorem}

\begin{proof}
Fix $U\subseteq[n]$, write $k=|U|$, and consider a star graph with one
internal vertex connected directly to every terminal leaf.  Give the edge
ending on each leaf in $U$ unit weight, the edges ending on the central and
purifier leaves weights $w_0,w_p>0$, and every other leaf edge weight
$\delta>0$.  Apply
\eqref{eq:centered-Q} before taking $\delta\downarrow0$.

The limiting cut for $T$ depends only on $x=|T\cap U|$:
\begin{equation}
 S(0T)=\min\{w_0+x,w_p+k-x\}
 =w_0+x-2(x-t)_+,
 \qquad t=\frac{w_p+k-w_0}{2}.
 \label{eq:star-hinge-short}
\end{equation}
Here $(y)_+:=\max\{y,0\}$.  Adding a sufficiently large common shift to
$w_0,w_p$ leaves $t$ unchanged,
so every $t\in\mathbb R$ is available with positive weights.  Balance removes
$w_0+x$, and positivity of the remaining positive-part term gives
\begin{equation}
 F_U(t):=\sum_Tc_T(|T\cap U|-t)_+\leq0.
 \label{eq:star-test-short}
\end{equation}
Set $h_t(r)=(r-t)_+$ and
$a_\ell(t)=\Delta^\ell h_t(0)$, where
$\Delta h(r)=h(r+1)-h(r)$.  The finite-difference identity
$h_t(|T\cap U|)=\sum_{V\subseteq T\cap U}a_{|V|}(t)$ then gives
\begin{equation}
 F_U(t)=\sum_{V\subseteq U}a_{|V|}(t)m_V.
 \label{eq:newton-star-short}
\end{equation}
For $U=\{i,j\}$ and $0<t<1$, balance removes the empty and singleton terms.
Since $a_2(t)=t$, \eqref{eq:star-test-short} gives
$t\,m_{ij}\leq0$.

It remains to exclude a nonzero functional with all $m_{ij}=0$.
Because $m_U=\sum_{T\supseteq U}c_T$, the inverse finite-subset transform is
$c_T=\sum_{U\supseteq T}(-1)^{|U|-|T|}m_U$; this is M\"obius inversion
on the subset lattice.  We may therefore choose a set $U$ of minimal size
with $m_U\ne0$.  Balance and the assumed vanishing of all pair
coefficients imply $k=|U|\geq3$.  Minimality reduces
\eqref{eq:newton-star-short} to $F_U(t)=a_k(t)m_U$.

The coefficient $a_k(t)$ changes sign between the following two choices:
\begin{equation}
 a_k(t)=(-1)^k t\quad(0<t<1),
 \qquad a_k(2)=(-1)^k(2-k).
 \label{eq:newton-sign-change}
\end{equation}
For $k\geq3$, the two values have opposite signs, so a nonzero $m_U$ is
incompatible with \eqref{eq:star-test-short} at both values of $t$.
Therefore some $m_{ij}$ is nonzero and, by
\eqref{eq:pair-sign-short}, strictly negative.
\end{proof}

The same conclusion holds for nonnegative linear combinations of balanced
coefficient arrays whose associated functionals are valid on graph entropy
vectors.

\paragraph{Quadratic response.}

The coefficient theorem is only half of the argument.  We now translate it
into a statement about entanglement response using a one-sided $C^2$
expansion and SSA.  This step should be distinguished from positivity of the
full HRT inequality at finite deformation amplitude: it constrains the local
boundary variation even when a
covariant proof of the complete inequality requires additional geometric
input.

We use the disjoint-support profiles \eqref{eq:legal-bump-family}, with $f_i=R g_i$.  Hold $R$ fixed and vary only the outward amplitudes.  For every
subset appearing in $Q_c$, assume that
$S[\gamma_0+\sum_i\lambda_i f_i]$ admits the one-sided $C^2$ expansion
\begin{align}
 S\!\left[\gamma_0+\sum_i\lambda_i f_i\right]
 ={}&S[\gamma_0]+\sum_i\lambda_i DS(f_i) \notag\\
 &+\frac12\sum_{i,j}\lambda_i\lambda_jD^2S(f_i,f_j)
 +o(|\lambda|^2),
 \label{eq:shape-Taylor-short}
\end{align}
near $\lambda=0$, with $\lambda_i\geq0$.  Nonnegative amplitudes keep the
deformations within the family constructed in
\eqref{eq:legal-bump-family}.

With a common amplitude,
\begin{equation}
 Q_c(\epsilon)=\sum_Tc_T
 S\!\left[\gamma_0+\epsilon\sum_{i\in T}f_i\right].
 \label{eq:bump-Q-short}
\end{equation}
Balance cancels the constant, linear, and diagonal quadratic terms, leaving
\begin{equation}
 Q_c''(0^+)=2\sum_{i<j}m_{ij}D^2S(f_i,f_j).
 \label{eq:bump-Hessian-short}
\end{equation}
Only the pair coefficients $m_{ij}$ remain.

For two disjoint outward deformations, SSA gives
\begin{equation}
 S[\gamma_0+\epsilon f_i]+S[\gamma_0+\epsilon f_j]
 -S[\gamma_0]-S[\gamma_0+\epsilon(f_i+f_j)]\geq0.
 \label{eq:bump-SSA-short}
\end{equation}
Define the pair susceptibility, or mixed response,
\begin{equation}
 \chi_{ij}:=-D^2S(f_i,f_j).
 \label{eq:pair-chi-short}
\end{equation}
Equivalently, it integrates the bilocal kernel
$-\delta^2S/[\delta\gamma(\phi)\delta\gamma(\phi')]$, often called the
off-diagonal entanglement density, against $f_i$ and $f_j$
\cite{Nozaki:2013entanglementdensity,Faulkner:2015shape}.
Equation~\eqref{eq:bump-SSA-short} gives $\chi_{ij}\geq0$.  Together with
Theorem~\ref{thm:strict-pair}, this gives
\begin{equation}
 \boxed{\frac12Q_c''(0^+)
 =\sum_{i<j}(-m_{ij})\chi_{ij}\geq0.}
 \label{eq:SSA-floor-short}
\end{equation}

Thus the quadratic response depends on $c_T$ only through $m_{ij}$.  It can
vanish only if $\chi_{ij}=0$ whenever $m_{ij}<0$.  This does not eliminate
higher-party information from the finite inequality; it only shows that the
quadratic deformation does not detect it.

For independent outward amplitudes $\lambda_i\geq0$, define
\begin{equation}
 \widehat Q_c(\lambda)
 =\sum_Tc_T S\!\left[\gamma_0+
       \sum_{i\in T}\lambda_i f_i\right].
 \label{eq:independent-bump-Q}
\end{equation}
At the origin, the diagonal entries vanish and, for $i\ne j$,
\begin{equation}
 \left.\partial_{\lambda_i}\partial_{\lambda_j}
 \widehat Q_c\right|_{0^+}
 =(-m_{ij})\chi_{ij}\geq0.
 \label{eq:mixed-bump-constraint}
\end{equation}
Hence the Hessian $H$ is copositive on the cone of outward amplitudes:
$\lambda^TH\lambda\geq0$ whenever every $\lambda_i\geq0$.  We make no statement
about signed directions, since inward deformations need not admit a common
set of boundary regions.  This is the infinitesimal counterpart of the finite
majorization ordering for the corresponding common-light-cone configurations
\cite{Grimaldi:2025hec,Grimaldi:2026com,Czech:2026maj}.

\section{Fixed-atom obstructions to RG monotonicity}
\label{sec:obstructions}

Here an \emph{atomization} is a decomposition into mutually disjoint
elementary boundary regions whose unions form every entropy argument.  By
\emph{fixed-atom} we mean that the same finite decomposition and the same
dimensionless profiles are used as the scale varies, although the complete
configuration may dilate self-similarly.

The quadratic projection found above does not imply that every finite
holographic inequality is of Shannon type, meaning that it can be written as
a nonnegative linear combination of the elemental monotonicity and
submodularity inequalities; it states only what this particular local probe
retains.  We begin with a six-party quantity whose coefficient array contains
genuinely non-Shannon information even though its quadratic
response is pairwise.  We then test two natural algebraic strategies for
recovering the discarded data, namely deleting parties and comparing a fixed
atomization at several scales.  Finally, we reproduce the pair projection
directly from extremal surfaces in an AdS$_4$ domain wall.

\subsection{The six-party test}
\label{sec:finite-g6}

The six-party quantity $G_6$ separates the information in the finite
inequality from the information retained by the local response.  Its full
coefficient array cannot be written as a nonnegative combination of elemental
Shannon-inequality coefficient arrays, whereas its quadratic response
contains only four pair susceptibilities.  All twelve entropy terms admit a
finite-amplitude realization on one common light cone.  Positivity is direct
for the weighted graph entropy models defined above.  The covariant HRT
statement requires the additional common-slice hypothesis stated below.
Under the endpoint assumptions, the self-similar gap also vanishes at both
conformal fixed points.
Thus $G_6$ cleanly separates the genuinely higher-party information present
in the finite inequality from the pair data retained by the quadratic probe
of Section~\ref{sec:pair-rigidity}.

A known six-party holographic entropy inequality
\cite{BaoNaskar:2024cm} provides the coefficient array:
\begin{align}
&S(ABC)+S(CDE)+S(BDF)+S(AEF)+S(ACE)+S(CDF)
  +S(ACD)+S(BCEF)
\notag\\
&\qquad \geq
 S(ACDE)+S(ABCEF)+S(BCDF)+S(CD)+S(AC)+S(B)
\notag\\[-1mm]
&\hspace{42mm}{}+S(CE)+S(A)+S(EF)+S(DF).
\label{eq:g6-catalogued-parent}
\end{align}
The parent inequality is \emph{superbalanced}: every single label and every
pair of labels occur with the same total multiplicity on the two sides.  Its
null reduction on $C$, obtained by retaining only the terms containing $C$,
is therefore again valid \cite{He:2020xuo,Grimaldi:2026com}:
\begin{align}
&S(ABC)+S(CDE)+S(ACE)+S(CDF)+S(ACD)+S(BCEF)
\notag\\
&\qquad \geq
S(ACDE)+S(ABCEF)+S(BCDF)+S(CD)+S(AC)+S(CE).
\label{eq:g6-null-reduced-parent}
\end{align}
Taking $C$ as the distinguished center and relabeling the remaining parties by
\begin{equation}
 C\mapsto0,\qquad A\mapsto a,\qquad B\mapsto b,\qquad
 D\mapsto c,\qquad E\mapsto d,\qquad F\mapsto e,
\label{eq:g6-relabeling}
\end{equation}
gives
\begin{equation}
\boxed{
\begin{aligned}
G_6:={}&S(0ab)+S(0ac)+S(0ad)+S(0cd)+S(0ce)+S(0bde)
\\
&-S(0a)-S(0c)-S(0d)-S(0acd)-S(0bce)-S(0abde).
\end{aligned}}
\label{eq:g6-centered-short}
\end{equation}
This algebraic reduction fixes the coefficient array.  The direct positivity
argument below does not invoke a physical limit of the terms removed from
\eqref{eq:g6-catalogued-parent}.  Balance follows because every label has the
same multiplicity on both sides.  For the ordinary labels,
\begin{equation}
 (n_a,n_b,n_c,n_d,n_e)=(3,2,3,3,2),
\label{eq:g6-occurrence-vector-short}
\end{equation}
while the center $0$ occurs six times on each side.

The coefficient array nevertheless contains information beyond the elemental
Shannon inequalities.  In the polymatroid language, one works with normalized
set functions, $r(\varnothing)=0$, while the elemental inequalities express
monotonicity, $r(U)\leq r(V)$ for $U\subseteq V$, and submodularity,
$r(U)+r(V)\geq r(U\cap V)+r(U\cup V)$, the set-function form of SSA.  Any
nonnegative sum of the latter inequalities is nonnegative on every normalized,
monotone, submodular rank function, so a single negative value on such a
function rules out a Shannon-type decomposition.

A concrete separating example is the rank function of the rank-four V\'amos
matroid \cite{Dougherty:2007vamos}.  A matroid is a finite abstraction of
linear independence, and its rank $r(X)$ is the maximum size of an independent
subset of $X$.  The V\'amos matroid is a standard eight-element example of
total rank four which cannot be represented by vectors over any field.  Group
its eight elements into the two-element blocks $a,b,c,d$, with ranks
\begin{align}
 r(a)=r(b)=r(c)=r(d)&=2,
 \label{eq:g6-vamos-singletons-short}\\
 r(ab)=r(ac)=r(ad)=r(bd)=r(cd)&=3,
 \qquad r(bc)=4.
 \label{eq:g6-vamos-pairs-short}
\end{align}
Every union of at least three blocks has rank four.  Adjoin $0$ and $e$ as
\emph{loops}, meaning rank-zero elements whose inclusion never changes $r$.
Substitution into \eqref{eq:g6-centered-short} gives
\begin{align}
 G_6[r]
 &=(3+3+3+3+2+3)-(2+2+2+4+4+4)
 \notag\\
 &=17-18=-1.
 \label{eq:g6-vamos-gap-short}
\end{align}
The positive-coefficient terms sum to $17$ and the negative-coefficient terms
to $18$.  Thus this rank function is a \emph{witness} in the precise sense
that it obeys every elemental Shannon constraint but gives $G_6[r]=-1$.
It therefore certifies that $G_6$ has no Shannon-type decomposition.  This is
only an abstract polymatroid test point, not an entropy vector realized by the
light-cone regions below.  Whether nonnegative
combinations of SSA generate $G_6$ on that physical family remains open.

A separate finite-amplitude construction realizes the twelve entropy terms
geometrically.  Choose a smooth positive base profile
$\gamma_0^{(R)}$ and five smooth nonnegative bumps $f_i^{(R)}$,
$i\in\{a,b,c,d,e\}$, with pairwise disjoint angular supports:
\begin{equation}
 f_i^{(R)}(\phi)\geq0,\qquad
 \operatorname{supp}f_i^{(R)}\cap
 \operatorname{supp}f_j^{(R)}=\varnothing
 \quad(i\neq j).
\label{eq:g6-disjoint-supports-short}
\end{equation}
For each subset of ordinary labels, define
\begin{equation}
 \gamma_T^{(R)}(\phi)
 =\gamma_0^{(R)}(\phi)+\sum_{i\in T}f_i^{(R)}(\phi).
\label{eq:g6-cut-family-short}
\end{equation}
At most one bump is active on each null generator.  The active subsets
therefore obey the chain condition of Theorem~\ref{thm:fiberwise-chain}, and
one set of disjoint atoms realizes all twelve terms at finite amplitude.
Following Section~\ref{sec:legality}, the unions are formed first on
$\Sigma_{\eta,R}$ and then taken to the balanced null limit
\eqref{eq:controlled-null-gap-short}.

Common atomization establishes geometric compatibility but not the
inequality.  For a weighted graph, positivity can be established by a
\emph{contraction certificate}.  This is a map between the bit strings which
encode membership in the positive- and negative-side cuts.  Its terminal
conditions require the occurrence bit string of each boundary terminal,
including the purifier, to map to the prescribed string on the other side.  Contraction
means that the map does not increase the coefficient-weighted Hamming
distance, namely the total weight of the coordinates on which two strings
differ.  Such a map glues pieces of minimum cuts for one side into admissible
comparison cuts for the other side with no greater total capacity.

\begin{proposition}[Six-party graph inequality]
\label{prop:g6-graph-inequality-short}
Let $0,a,b,c,d,e$ label disjoint boundary parties, represented by terminal
vertices of a single weighted graph whose minimum cuts reproduce the twelve
entropies in \eqref{eq:g6-centered-short}.  Then
\begin{equation}
 G_6\geq0.
 \label{eq:g6-graph-positive-short}
\end{equation}
The same conclusion holds for RT entropies when the relevant surfaces are
simultaneously globally minimizing on one Riemannian bulk slice
\cite{Ryu:2006bv}.
\end{proposition}

\begin{proof}
An explicit contraction map between the six-bit cut-membership strings for
the positive and negative terms obeys the required terminal assignments and
is $1$-Lipschitz in the corresponding Hamming metrics.  The
proof-by-contraction theorem \cite{Bao:2015bfa,BaoNaskar:2024cm} then gives
the result.
\end{proof}

The graph result applies whenever one weighted graph reproduces all twelve
entropies.  The RT corollary similarly requires the relevant surfaces to
minimize on one Riemannian bulk slice.  A common boundary atomization does not
ensure a common bulk Cauchy slice for covariant HRT surfaces.  Here the
common-slice hypothesis is that the six negative-coefficient HRT surfaces
minimize simultaneously on one bulk Cauchy slice.  The positive-coefficient
surfaces need only admit comparison surfaces on that slice, supplied by Wall's
maximin construction.  This construction obtains an HRT surface by minimizing
area on each bulk Cauchy slice and then maximizing that minimum over slices
\cite{Wall:2012uf}.  Because $0acd,0bce,0abde$ are crossing regions -- each pair
overlaps but is neither nested nor disjoint -- this hypothesis is stronger than
the usual common-slice result for nested or disjoint regions
\cite{Rota:2017ubr}, but weaker than simultaneous minimization of all twelve
surfaces.  Under this hypothesis, cut-and-paste of the positive-side
comparison surfaces gives candidates for the negative-side regions with no
larger total area, and hence $G_6\geq0$.  Weighted min-cut graphs give the
alternative in
Proposition~\ref{prop:g6-graph-inequality-short}
\cite{GradoWhite:2024gtx,GradoWhite:2025zqw}.

A self-similar family turns the gap into a function of scale without changing
the dimensionless shape:
\begin{equation}
 \gamma_0^{(R)}(\phi)=R\,g_0(\phi),
 \qquad
  f_i^{(R)}(\phi)=R\,f_i^{(1)}(\phi),
\label{eq:g6-self-similar-short}
\end{equation}
with fixed dimensionless profiles, and set
$G_6(R):=G_6^{\rm null}(R)$.  The function is nonnegative at every scale where
the common-slice and regulated null-limit assumptions above hold.

The quadratic response retains only four entries of the non-Shannon
coefficient array.  The ten pair sums $m_{ij}$ are
\begin{center}
\begin{tabular}{crrrrrrrrrr}
\toprule
pair & $ab$&$ac$&$ad$&$ae$&$bc$&$bd$&$be$&$cd$&$ce$&$de$\\
\midrule
$m_{ij}(G_6)$&$0$&$0$&$-1$&$-1$&$-1$&$0$&$-1$&$0$&$0$&$0$\\
\bottomrule
\end{tabular}
\end{center}
Only $ad$, $ae$, $bc$, and $be$ are nonzero.  Scale all five bumps by a common
outward amplitude $\epsilon\geq0$.  Assume that, for sufficiently small $\epsilon\geq0$, the twelve entropies
remain on a single smooth extremal-surface phase, or branch, with no change
of minimizing topology.  We further assume that each admits a $C^2$
expansion from the right: writing $S_T(\epsilon)$ for any one of the twelve
entropy terms,
\[
 S_T(\epsilon)
 =S_T(0)+\epsilon S_T'(0^+)
 +\frac12\epsilon^2 S_T''(0^+)+o(\epsilon^2),
 \qquad \epsilon\to0^+.
\]
Here ``from the right'' means that $\epsilon$ approaches zero through
positive values, corresponding to outward deformations; no extension to
$\epsilon<0$ is assumed.  Then Eq.~\eqref{eq:SSA-floor-short} gives
\begin{equation}
 \frac12G_6''(0^+)
 =\chi_{ad}+\chi_{ae}+\chi_{bc}+\chi_{be}\geq0.
 \label{eq:g6-local-floor-short}
\end{equation}
Each $\chi_{ij}$ is an SSA-controlled pair susceptibility.  No other
coefficient data enter the quadratic response.  The derivatives in
\eqref{eq:g6-local-floor-short} are with respect to the bump amplitude
$\epsilon$, not the RG scale $R$.

The full gap vanishes at a conformal fixed point.  In three boundary
dimensions, the common-light-cone entropy has the form
\cite{Casini:2017vbe,Casini:2017roe,Casini:2018kzx}
\begin{equation}
 S_*[\gamma]=\alpha_\delta\,\mathcal P[\gamma]-F_*.
\label{eq:g6-cft-entropy-short}
\end{equation}
Here $\mathcal P[\gamma]$ is the regulator-dependent local perimeter, or
cut, functional.  For the standard common-light-cone regulator, it is
additive over null generators and affine, meaning linear up to a constant, in
the cut profile.  The multiplicities in
\eqref{eq:g6-occurrence-vector-short} cancel this local term generator by
generator, while the six terms on each side cancel $F_*$.  Hence
\begin{equation}
 G_6\big|_{\mathrm{CFT}}=0.
\label{eq:g6-cft-zero-short}
\end{equation}
The endpoint limits require an additional convergence assumption: the
renormalized gap for \eqref{eq:g6-self-similar-short} must approach the
corresponding light-cone entropy functionals of the UV and IR CFTs on
connected boundary cuts, with
sufficient uniformity to exchange the null and endpoint limits.  Then
\begin{equation}
 \lim_{R\to0}G_6(R)=0,
 \qquad
 \lim_{R\to\infty}G_6(R)=0.
\label{eq:g6-endpoints-short}
\end{equation}
If the common-slice hypothesis holds throughout the flow, $G_6(R)$ is
nonnegative at every scale and vanishes at both endpoints.  Any monotone
function with these properties must vanish identically.  Thus $G_6$ can
provide a nonnegative fixed-scale constraint but cannot itself furnish a
nontrivial one-sided RG monotone under these endpoint assumptions.  This
closes the most direct attempt to promote the finite gap into a flow function.

\subsection{Deletion and finite scale differences}
\label{sec:remove-pair}

The $G_6$ example contains finite coefficient data invisible to the quadratic
response.  Two finite constructions might retain this information.  The first
deletes selected parties and asks whether a genuinely multipartite inequality
survives after the pair contribution disappears.  The second combines
nonnegative graph inequalities and conditional mutual informations at several
scales in search of a finite difference $M(R)-M(\lambda R)$.  Both questions
are algebraic statements about entropy coefficients; an application to HRT
entropies requires separate geometric assumptions.  We begin with exact
deletion and then turn to the scale construction.

\paragraph{Deleting parties.}

Deletion means an exact identity of the coefficient array: the linear
combination must vanish after one or more ordinary parties are erased, not
only on a chosen state or entropy trajectory.  On the classical holographic
graph cone, this condition leaves no inequality beyond strong subadditivity.  One
or two independent deletions give nonnegative sums of conditional mutual
informations, while three force the coefficient array to vanish.

For a set of ordinary labels $L$, encode a coefficient array with common
center $0$ in the multilinear polynomial
\begin{equation}
 \mathcal C(z)=\sum_{T\subseteq L}c_Tz^T,
 \qquad z^T=\prod_{i\in T}z_i,
 \label{eq:coefficient-polynomial-short}
\end{equation}
where $0$ is implicit.  Erasing $p$ identifies $S(0Up)$ with $S(0U)$ and
therefore sets $z_p=1$.  With $C=L\setminus\{p\}$, exact deletion is
equivalent to
\begin{equation}
 \mathcal C(z)=(1-z_p)\mathcal B(z_{L\setminus\{p\}}).
 \label{eq:one-deletion-factor-short}
\end{equation}
The factor pairs terms that differ only by $p$.  We use the convention
$\mathcal B(z)=-\sum_{U\subseteq C}b_Uz^U$.  Since
\eqref{eq:one-deletion-factor-short} is a polynomial identity, it is stronger
than cancellation on a single entropy branch.

For one deleted party, every such linear combination has the form
\begin{equation}
 Q_b[S]=\sum_{U\subseteq C}b_U
 \bigl[S(0Up)-S(0U)\bigr].
 \label{eq:one-deletion-form-short}
\end{equation}
The bracket records the entropy change after adding $p$ to $0U$.  The
classification of combinations that are nonnegative on every graph entropy
vector uses the Boolean lattice $2^C$, namely the subsets of $C$ ordered by
inclusion.  An upper set
$\mathcal F\subseteq2^C$ contains every superset of each element.  Orient each
covering edge $U\to Ui$ upward, where $Ui:=U\cup\{i\}$, and for edge
weights $w$ define the net outward flow at $U$ by
\begin{equation}
 (\operatorname{div}w)_U
 =\sum_{i\notin U}w_{U,i}
 -\sum_{i\in U}w_{U\setminus\{i\},i}.
 \label{eq:boolean-div-short}
\end{equation}
Thus $w_{U,i}>0$ sends flow from $U$ to $Ui$.  Graph positivity is equivalent
to a nonnegative upward flow, which in entropy variables is a nonnegative sum
of conditional mutual informations.

\begin{theorem}[One-deletion classification]
\label{thm:one-deletion-short}
For \eqref{eq:one-deletion-form-short}, the following are equivalent:
\begin{enumerate}[label=(\roman*)]
 \item $Q_b\geq0$ on every classical holographic graph entropy vector;
 \item
 \begin{equation}
  \sum_{U\subseteq C}b_U=0,
  \qquad
  \sum_{U\in\mathcal F}b_U\leq0
  \quad\text{for every upper set }\mathcal F;
  \label{eq:upper-set-short}
 \end{equation}
 \item $b=\operatorname{div}w$ for a nonnegative upward flow;
 \item
 \begin{equation}
  Q_b=\sum_{U\subseteq C}\sum_{i\notin U}
  w_{U,i}I(p:i\mid0U),
  \qquad w_{U,i}\geq0.
  \label{eq:one-deletion-CMI-short}
 \end{equation}
\end{enumerate}
\end{theorem}

\begin{proof}
Set $D_U=S(0Up)-S(0U)$.  Along an upward edge $U\to Ui$, the change in
$D_U$ is
\begin{equation}
 D_U-D_{Ui}=I(p:i\mid0U),
 \label{eq:deletion-edge-CMI-short}
\end{equation}
which is nonnegative by SSA.  Boolean-lattice summation by parts turns
$b=\operatorname{div}w$ into \eqref{eq:one-deletion-CMI-short}, proving
(iii)$\Rightarrow$(iv)$\Rightarrow$(i).

Conversely, Appendix~\ref{app:selectors} constructs for each upper set
$\mathcal F$ a positive-weight graph satisfying
\begin{equation}
 S(0Up)-S(0U)=1-2\mathbf1_{\mathcal F}(U).
 \label{eq:upper-selector-short}
\end{equation}
The empty and full upper sets imply $\sum_Ub_U=0$, and every other upper set
gives its inequality in \eqref{eq:upper-set-short}.  Hence
(i)$\Rightarrow$(ii).  The max-flow/min-cut argument in the same appendix
equates these inequalities with the existence of a nonnegative upward flow,
so (ii)$\Rightarrow$(iii).
\end{proof}

Thus any linear entropy combination that is nonnegative on every graph
entropy vector and has one exact deletion factor is a nonnegative sum of
conditional mutual informations on Boolean-lattice edges.  One deletion
therefore supplies no additional building block for the quadratic
construction.  This is a special case of inclusion dominance
\cite{Grimaldi:2026com}: a single map sends every selection $U$ of
left-hand-side terms to an equal-cardinality selection $f(U)$ on the right
which contains each party at least as often, and preserves inclusions,
$U'\subseteq U\Rightarrow f(U')\subseteq f(U)$.

A second deletion also leaves no intrinsically multipartite remainder.
Vanishing after either $p$ or $q$ is erased gives two independent factors.
After expansion, each remaining coefficient multiplies one conditional mutual
information, and graph positivity fixes its sign separately.

\begin{theorem}[Double-deletion classification]
\label{thm:double-deletion-short}
Suppose a coefficient array with common center $0$ vanishes algebraically
after erasing either $p$ or $q$.  With $C=L\setminus\{p,q\}$, its polynomial
has the form
\begin{equation}
 \mathcal C(z)=(1-z_p)(1-z_q)
 \sum_{U\subseteq C}d_Uz^U.
 \label{eq:double-factor-short}
\end{equation}
It is nonnegative on every classical holographic graph entropy vector if and
only if
\begin{equation}
 Q=\sum_{U\subseteq C}w_UI(p:q\mid0U),
 \qquad w_U=-d_U\geq0.
 \label{eq:double-CMI-short}
\end{equation}
\end{theorem}

\begin{proof}
Expanding the contribution of a fixed monomial $d_Uz^U$ in
\eqref{eq:double-factor-short} gives $-d_UI(p:q\mid0U)$, establishing the
decomposition.  It remains to fix the sign of each weight.

Choose any $U_*\subseteq C$.  The star graph constructed in
Appendix~\ref{app:selectors} turns on only the conditional mutual information
with conditioning set $U_*$:
\begin{equation}
 I(p:q\mid0U)=2\epsilon\,\delta_{U,U_*}.
 \label{eq:isolated-CMI-short}
\end{equation}
Evaluation on this graph forces $w_{U_*}\geq0$.  Since $U_*$ is arbitrary,
every weight is nonnegative.  Conversely, SSA makes
\eqref{eq:double-CMI-short} nonnegative for nonnegative weights.
\end{proof}

A third independent deletion is more restrictive.  Add the third factor
$1-z_r$, apply Theorem~\ref{thm:double-deletion-short} to $p,q$, and
treat $r$ as a remaining label.  Multiplication by $1-z_r$ pairs each
two-deletion CMI coefficient with its negative.  Graph positivity requires
both coefficients to be nonnegative, so all vanish.  A nonzero coefficient
array that is nonnegative on every graph entropy vector can therefore have at
most two independent deletion factors.

In summary, one deletion gives a Boolean-lattice sum of conditional mutual
informations, two give \eqref{eq:double-CMI-short}, and three leave only zero.
Exact deletion therefore produces no graph inequality free of the pair
structures controlled by SSA.  This result does not cover cancellation on a
particular entropy branch: all deletion conditions above are identities of
the coefficient array.

\paragraph{Combining scales.}

Deletion acts at one scale.  The second construction evaluates the same
regions at several scales and seeks an identity of the form
$M(R)-M(\lambda R)$ with a nonnegative right-hand side.  Such an identity
would produce monotonicity under the finite step $R\mapsto\lambda R$.

The analysis is finite and algebraic.  The atomization remains fixed and
self-similar as $R$ varies.  Only finitely many integer scale assignments occur
in the nonnegative combination and in the candidate $M$; we call this property
\emph{finite scale support}.  The identity must also hold coefficient by
coefficient, not only along one RG trajectory.  The allowed nonnegative terms
are balanced graph inequalities and conditional mutual informations.  Under
these assumptions, no nonzero combination of the allowed terms is a finite
scale difference.

Fix $\lambda>1$ and assign an integer scale index to each of the $n$ ordinary
atoms,
$\mathbf h=(h_1,\ldots,h_n)\in\mathbb Z^n$.  The Laurent monomial
$x^{\mathbf h}=\prod_i x_i^{h_i}$, which permits positive or negative
integer powers, records these relative scale assignments.
For an entropy coefficient $c$ whose terms share the center $0$, define
\begin{equation}
 q_c(x)=\sum_{T\subseteq[n]}c_Tx^{\mathbf1_T},
 \qquad
 x^{\mathbf1_T}=\prod_{i\in T}x_i.
 \label{eq:scale-boolean-column-short}
\end{equation}
We call $x^{\mathbf h}q_c(x)$ the coefficient column at scale assignment
$\mathbf h$, meaning the same coefficient pattern $q_c$ placed at that
assignment; conditional mutual informations have the same representation.

Let $\gamma_{\mathbf h}(R)$ denote the shape with scale assignment
$\mathbf h$, and write $\mathbf1=(1,\ldots,1)$.  Self-similarity identifies a
dilation of the overall scale with a simultaneous unit shift:
\begin{equation}
 \gamma_{\mathbf h}(\lambda R)
   =\gamma_{\mathbf h+\mathbf1}(R).
 \label{eq:scale-grade-closure-short}
\end{equation}
Consequently, $R\mapsto\lambda R$ multiplies the Laurent polynomial by
\begin{equation}
 P = \prod_{i=1}^n x_i.
 \label{eq:diagonal-shift-short}
\end{equation}
Thus $P$ shifts every atom index by one.

Let the finite Laurent polynomial $B$ encode the candidate $M$.  Then $PB$
encodes the same combination one scale step later, and $(1-P)B$ its finite
difference.  The required coefficient identity is
\begin{equation}
 (1-P)B=A,
 \qquad
 A=\sum_\alpha w_\alpha x^{\mathbf h_\alpha}q_{c_\alpha}
   +\sum_\beta v_\beta x^{\mathbf k_\beta}q_{I_\beta},
 \qquad w_\alpha,v_\beta\geq0.
 \label{eq:scale-cocycle-short}
\end{equation}
Here $A$ is a nonnegative sum of balanced inequalities valid for graph
entropies, with coefficients $c_\alpha$, and conditional mutual informations,
with coefficients $I_\beta$, at finitely many scale assignments.  The
identity is imposed coefficient by coefficient, a stronger condition than
equality along one selected entropy trajectory.

On the hypersurface $P=1$, every finite difference $(1-P)B$ vanishes.  A
nonzero $A$ of the form above does not: one of its pair sums gives a
nonvanishing quadratic term.  Choose a pair $i,j$ and restrict the Laurent
variables to
\begin{equation}
 x_i=e^t,
 \qquad x_j=e^{-t},
 \qquad x_{k\ne i,j}=1.
 \label{eq:reciprocal-curve-short}
\end{equation}
The factors $e^t$ and $e^{-t}$ cancel in $P$, which is why we call this the
reciprocal curve; it lies in $P=1$.
Along it, a balanced coefficient at scale assignment $\mathbf h$ has the
expansion
\begin{equation}
 x^{\mathbf h}q_c(x)
   =-m_{ij}(c)t^2+O(t^3).
 \label{eq:translated-reciprocal-expansion-short}
\end{equation}
Balance removes the constant and linear terms of $q_c$, so multiplication by
$x^{\mathbf h}$ changes only terms of order $t^3$ and higher.  The quadratic
coefficient remains $-m_{ij}(c)$.  For $A$, it is $-m_{ij}(A)$, the weighted
sum of the corresponding pair contributions.

\begin{theorem}[Finite-scale no-go]
\label{thm:scale-no-go-short}
Let $A\ne0$ be a finite nonnegative combination of the translated
graph-inequality and conditional-mutual-information coefficient columns in
\eqref{eq:scale-cocycle-short}.  There is no finite Laurent polynomial $B$
satisfying $(1-P)B=A$.  The same conclusion holds when a scale step also
permutes the atoms by a permutation of finite order.
\end{theorem}

\begin{proof}
Assume first that a finite Laurent polynomial $B$ satisfies
\eqref{eq:scale-cocycle-short}.  Its finite difference $(1-P)B$ vanishes
identically whenever $P=1$.

The polynomial $A$ cannot vanish on the same hypersurface.  By
Theorem~\ref{thm:strict-pair}, one nonzero summand has a strictly negative
pair sum for some $i,j$, while every other summand has a nonpositive sum for
that pair.  The nonnegative weights cannot cancel the strict term, so
$m_{ij}(A)<0$.  The reciprocal curve
\eqref{eq:reciprocal-curve-short} lies in $P=1$, but along it $A$ obeys
\begin{equation}
 A(t)=-m_{ij}(A)t^2+O(t^3)>0,
 \label{eq:reciprocal-expansion-short}
\end{equation}
for sufficiently small nonzero $t$.  This is incompatible with the
vanishing of $(1-P)B$ on $P=1$.

A finite-order relabeling gives the same contradiction.  If a scale step also
applies $\sigma$ with $\sigma^m=1$, summing the identity over one orbit gives
\begin{equation}
 (1-P^m)B
   =\sum_{r=0}^{m-1}(P\sigma)^rA.
 \label{eq:orbit-cocycle-short}
\end{equation}
After collecting the finite orbit, the right-hand side remains a nonzero
nonnegative sum of the allowed terms and therefore has a strictly negative
pair sum.  The $P=1$ argument again gives a contradiction.
\end{proof}

The theorem applies only to exact coefficient identities with finite support,
fixed self-similar regions, and the terms in
\eqref{eq:scale-cocycle-short}.  It does not cover identities valid only along
one RG flow, infinitely many scales, regions that vary with $R$, or
relabelings of infinite order.

Finite support also constrains the endpoint value of a candidate $M$.  This
separate statement uses balance and finite support but not positivity.  Write
\begin{equation}
 B(x)=\sum_{\mathbf h}b_{\mathbf h}x^{\mathbf h},
 \qquad
 A(x)=(1-P)B(x)=\sum_{\mathbf h}a_{\mathbf h}x^{\mathbf h},
 \label{eq:grade-potential-short}
\end{equation}
and suppose that $A$ is a finite signed sum of balanced coefficient columns.
The total scale index $d(\mathbf h)=\sum_i h_i$ obeys, for every translated
inequality,
\begin{equation}
 \sum_Tc_Td(\mathbf h+\mathbf1_T)
 =d(\mathbf h)\sum_Tc_T+
   \sum_i\sum_{T\ni i}c_T=0.
 \label{eq:balanced-grade-short}
\end{equation}
Thus balance makes the coefficient sum weighted by $d(\mathbf h)$ vanish.
The same holds for finite signed sums.  Multiplication by $P$ raises every
atom index by one and hence increases $d$ by $n$.  For
$A=(1-P)B$ with finite-support $B$, this gives
\begin{equation}
 0=\sum_{\mathbf h}d(\mathbf h)a_{\mathbf h}
   =-n\sum_{\mathbf h}b_{\mathbf h}.
 \label{eq:grade-coboundary-short}
\end{equation}
A finite permutation only reorders the indices and preserves their sum.
Whether or not the scale step includes such a permutation, the coefficients
of $B$ therefore satisfy
\begin{equation}
 \sum_{\mathbf h}b_{\mathbf h}=0.
 \label{eq:zero-charge-short}
\end{equation}
Thus $B$ has zero total coefficient.  An endpoint statement requires
additional physical assumptions. For a three-dimensional CFT vacuum, the entropy of a connected
common-light-cone cut has the regulated form
\begin{equation}
 S_*[\gamma]=\alpha_\delta\mathcal P[\gamma]-F_*.
 \label{eq:fixed-point-affine-short}
\end{equation}
Here $\mathcal P$ is the regulator-dependent local cut contribution,
evaluated with the same regulator and counterterm subtraction prescription in
every term
\cite{Casini:2017vbe,Casini:2017roe,Casini:2018kzx}.  Assume that the
self-similar cuts cancel it:
$\sum_{\mathbf h}b_{\mathbf h}\mathcal P[\gamma_{\mathbf h}]=0$.  Equation
\eqref{eq:zero-charge-short} then cancels the universal constant $F_*$ as
well.

Define
\begin{equation}
 M_B(R)=\sum_{\mathbf h}b_{\mathbf h}S[\gamma_{\mathbf h}(R)],
 \label{eq:potential-functional-short}
\end{equation}
and assume convergence to the corresponding light-cone entropy functionals on
connected boundary cuts at both the UV and IR fixed points.  Then
\begin{equation}
 \lim_{R\to0}M_B(R)=
 \lim_{R\to\infty}M_B(R)=0.
 \label{eq:potential-endpoints-short}
\end{equation}

If $M_B(R)\geq M_B(\lambda R)$ at every $R$, iteration toward the IR gives
$M_B(R)\geq0$, while backward iteration toward the UV gives $M_B(R)\leq0$.
Thus $M_B$ vanishes identically.

This endpoint conclusion requires cancellation of the local terms,
convergence to the corresponding light-cone entropy functionals on connected
boundary cuts at both fixed points as in
\eqref{eq:potential-endpoints-short}, and monotonicity at every scale.
Without these assumptions, \eqref{eq:zero-charge-short} alone implies neither
the endpoint limits nor the vanishing of $M_B(R)$.

\subsection{A bulk check of pair rigidity}
\label{sec:bulk-response}

The coefficient and endpoint arguments above close the finite fixed-atom
route to a nontrivial monotone under their stated assumptions.  However, they
do not explain why the same pair projection should arise geometrically.  We
therefore apply a special conformal transformation (SCT) in a null direction
transverse to the undeformed extremal surface.  Its parameter $\kappa$ is
distinct from both the bump amplitude and the RG scale
\cite{Nozaki:2013entanglementdensity,Faulkner:2015shape,Casini:2015sumrule}.
On the regulated \emph{noncompact branch} -- where the open extremal surface is
truncated at a remote cutoff held fixed while differentiating in $\kappa$ -- the
order-$\kappa^2$ response contains only constant, one-bump, and two-bump terms.
The remote cutoff is removed only afterward.  Balance
removes the first two,
while the remaining pair terms admit a nonnegative bulk representation which
is strict away from exact AdS.  This is an independent geometric check of
pair rigidity, not an RG monotonicity theorem.

We work with an asymptotically AdS$_4$ domain wall in conformal radial gauge,
\begin{equation}
 ds^2=a(z)\bigl(-dx^+dx^-+dy^2+dz^2\bigr),
 \qquad a(z)=h(z)^{-2},
 \label{eq:dw-metric}
\end{equation}
with $h>0$ and $C^3$ in the interior, and
\begin{equation}
 h(z)=\frac{z}{L_{\rm UV}}+o(z),
 \label{eq:bulk-uv-h}
\end{equation}
at the conformal boundary.  Einstein's equations and the null energy
condition imply
\begin{equation}
 h''(z)\geq0.
 \label{eq:dw-nec}
\end{equation}
Thus the null energy condition (NEC) makes $h$ convex.  No other property of
the matter sector enters the sign argument
\cite{Freedman:1999gp,Myers:2010xs}.

We work on the regular noncompact branch of connected extremal surfaces.
Both the UV and remote cutoffs are held fixed during differentiation and
removed afterward under the falloff assumptions stated below.

Introduce null coordinates
\begin{equation}
 U=x^+,
 \qquad V=x^-,
 \label{eq:bulk-UV-convention}
\end{equation}
so that the undeformed surface lies at $V=0$.  The linear bulk extension of a
boundary cut $f_T(y)$ is the solution $F_T$ of
\begin{equation}
 \mathcal D F_T=0,
 \qquad
 \mathcal D:=\partial_z(a\partial_z)+\partial_y(a\partial_y).
 \label{eq:D-operator}
\end{equation}
Here $\mathcal D$ is the Jacobi operator, namely the linearized
extremal-surface operator about the undeformed graph.  Let $f_0$ be a common
cut and let $f_i,f_j$ be
nonzero, nonnegative boundary bumps with disjoint closed supports.  Linearity
gives
\begin{equation}
 f_T=f_0+\sum_{\ell\in T}f_\ell,
 \qquad
 F_T=F_0+\sum_{\ell\in T}F_\ell,
 \qquad F_\ell(y,0)=f_\ell(y).
 \label{eq:additive-cuts}
\end{equation}
Although the boundary supports are disjoint, their bulk extensions overlap.
The strong maximum principle gives
\begin{equation}
 F_i>0,
 \qquad F_j>0,
 \label{eq:positive-jacobi}
\end{equation}
in the connected interior.

We assume that the extremal surface is uniquely representable as a
single-valued graph over $(y,z)$, that the Dirichlet Jacobi problem has no zero
mode, meaning no nontrivial solution with vanishing boundary data, and that
the fields obey the regulator falloff conditions in
Appendix~\ref{app:bulk-pair-response}.  These
conditions fix the branch, make the linearized boundary-value problem
nondegenerate, and control the boundary terms when the cutoffs are removed.

The SCT deformation has the following first-order transverse displacement:
\begin{equation}
 N(y,z)=-y^2+n(z),
 \qquad p(z):=n'(z),
 \qquad (ap)'=2a,
 \label{eq:N-definition}
\end{equation}
where $n(0)=0$.  Regularity at the endpoint selects
\begin{equation}
 p(z)=-2a(z)^{-1}\int_z^{z_{\rm IR}}a(s)\,ds<0,
 \qquad p(z)=-2z+o(z)\quad(z\to0).
 \label{eq:p-regular}
\end{equation}
Here $z_{\rm IR}$ is the regular deep-bulk endpoint of the radial interval,
possibly $+\infty$.  Thus regularity fixes both the sign and the near-boundary
behavior of $p$.

Let $\kappa$ be the SCT parameter, and let $S_T(\kappa)$ be the renormalized
entropy in the fixed domain-wall vacuum after the transformation.  Define the
conditional mutual information of the two bump regions, conditioned on the common
region,
\begin{align}
 I_\kappa(i:j\mid0)
 &=S_{\{i\}}(\kappa)+S_{\{j\}}(\kappa)
   -S_\varnothing(\kappa)-S_{\{i,j\}}(\kappa),
 \notag\\
 \mathcal K^{(\kappa)}_{ij,\mathrm{nc}}
 &:=\left.\frac12\partial_\kappa^2
 I_\kappa(i:j\mid0)\right|_{\kappa=0}.
 \label{eq:K-kappa-definition}
\end{align}
It removes the separate responses of $i$ and $j$.  The label ``nc'' denotes the
noncompact branch with the prescribed order of limits: both regulators stay
finite during differentiation, and the remote regulator is removed
afterward.  Appendix~\ref{app:bulk-pair-response} retains all boundary terms
and derives
\begin{equation}
 \mathcal K^{(\kappa)}_{ij,\mathrm{nc}}
 =\frac{1}{4G_N}\frac18
 \int_0^{z_{\rm IR}}dz\int_{\mathbb R}dy\,
 F_iF_j\,\mathcal Dg,
 \qquad
 g:=|\nabla N|^2=4y^2+p^2.
 \label{eq:K-positive-representation}
\end{equation}
The relation $(ap)'=2a$ gives
\begin{equation}
 \mathcal Dg
 =4a\left[
 \left(2+\frac{h'}h p\right)^2
 +\frac{h''}h p^2\right]\geq0.
 \label{eq:Dg-positive}
\end{equation}
The first term is a square and the second is nonnegative by the NEC.
Equation~\eqref{eq:positive-jacobi} then fixes the sign of every factor in
\eqref{eq:K-positive-representation}.

\begin{theorem}[SCT pair positivity]
\label{thm:local-pair-strictness}
Let \eqref{eq:dw-metric} be a $C^3$ asymptotically AdS Einstein-NEC
domain wall satisfying the assumptions above.  For two nonzero nonnegative
smooth bumps with disjoint closed supports,
\begin{equation}
 \mathcal K^{(\kappa)}_{ij,\mathrm{nc}}\geq0.
\end{equation}
The inequality is strict unless the connected bulk component probed by the
extremal surface is exact AdS.  Exact AdS saturates it.
\end{theorem}

\begin{proof}
Equation~\eqref{eq:positive-jacobi} makes $F_iF_j$ strictly positive in the
interior, while Eq.~\eqref{eq:Dg-positive} makes $\mathcal Dg$
nonnegative.  The integral in
Eq.~\eqref{eq:K-positive-representation} is therefore nonnegative.

Suppose that it vanishes.  Continuity and the strict positivity of $F_iF_j$
then force both nonnegative terms in \eqref{eq:Dg-positive} to vanish:
\begin{equation}
 h''=0,
 \qquad 2+\frac{h'}h p=0,
\end{equation}
throughout the component.  The UV normalization and endpoint regularity give
$h=z/L$ and $p=-2z$.  These are the exact-AdS values.  Conversely, for this
solution $\mathcal Dg$ vanishes identically, so the pair response vanishes
and saturates the bound.
\end{proof}

For any finite collection of disjoint bumps, the order-$\kappa^2$ response
is quadratic in the indicators that specify which bumps are present:
\begin{equation}
 \Xi_{\mathrm{nc}}(T):=
 \left.\frac12\partial_\kappa^2S_T(\kappa)\right|_{\kappa=0}
 =A+\sum_iA_i\mathbf1_{i\in T}
 +\sum_{i<j}A_{ij}\mathbf1_{i\in T}\mathbf1_{j\in T}.
 \label{eq:sct-degree-two}
\end{equation}
Here $\mathbf1_{i\in T}$ equals one when bump $i$ is present and zero
otherwise.  Because $F_T=F_0+\sum_{i\in T}F_i$ is a sum of the individual
linearized, or Jacobi, fields, a second-order variation contains at most two
such indicators.  Appendix~\ref{app:bulk-pair-response} makes this counting
precise.  Balance of the coefficients $c_T$ cancels the
constant and one-bump terms.  With pair sums
$m_{ij}=\sum_{T\supseteq\{i,j\}}c_T$, the full response is
\begin{equation}
 q_{2,\mathrm{nc}}[c]:=\sum_Tc_T\Xi_{\mathrm{nc}}(T)
 =\sum_{i<j}(-m_{ij})
 \mathcal K^{(\kappa)}_{ij,\mathrm{nc}}.
 \label{eq:centered-hei-local-response}
\end{equation}
All sums involving triples and larger subsets drop out of
\eqref{eq:centered-hei-local-response}.  If $c$ is nonzero and valid on
graph entropy vectors, Theorem~\ref{thm:strict-pair} gives
$-m_{ij}>0$ for at least one pair and $-m_{kl}\geq0$ for every pair $k,l$.
If both bumps in the strictly weighted pair are nonzero and the corresponding
connected bulk component probed by the extremal surface is not exact AdS,
$\mathcal K^{(\kappa)}_{ij,\mathrm{nc}}$ is strictly positive.  The total
response is then strictly positive.

This result applies only to the noncompact branch and the order of limits
specified above.  At finite $\kappa$, the compact SCT surface closes at
$\ell\sim|\kappa|^{-1}$, so its remote cutoff moves with $\kappa$ rather than
remaining fixed during differentiation.  Appendix~\ref{app:compact-sct-limit}
isolates the resulting \emph{cap term}, namely the correction to the
pair-subtracted entropy from the remote closing cap of the compact surface,
which is absent from the open fixed-cutoff branch.  The compact extension would follow from an
$o(\kappa^2)$ bound on that term, which remains unproved.

With this qualification, the bulk calculation completes the obstruction
analysis.  We now return to the boundary entropy cone and use two
constructions that evade its finite fixed-atom assumptions.

\section{Radial and angular constraints beyond fixed atoms}
\label{sec:beyond-ssa}

The preceding obstructions rely on two linked restrictions: the number of
parties is finite, and the atomization is held fixed as the scale varies.  We
now relax these assumptions one at a time.  The five-party construction lets
the distinguished region track an already-grown petal; a Markov equality then
removes the zeroth-order entropy gap and exposes a local growth bound.  The
odd-cyclic construction instead sends the number of angular parties to
infinity, avoiding Markov saturation at the price of a residual shape
response.  These mechanisms provide complementary partial answers to the
question posed in the Introduction.

\subsection{Five-party constraints from a moving shell}
\label{sec:radial}

We begin by expanding about a region which has already grown.  This region is
absorbed into the distinguished party of a five-party inequality, while only
the next radial shell is treated as a new outer atom.  The construction does
not contradict Theorem~\ref{thm:scale-no-go-short}: the atomization depends on
the scale, and the cancellation which makes the derivative finite is a
conditional-independence property of the chosen entropy trajectory rather
than an identity of coefficients.

All entropies in this subsection are evaluated on one regulated spacelike slice
approaching the light cone.  We work at leading holographic order and assume
that the five-party inequality is valid for the corresponding regulated HRT
entropies.  The same algebra is unconditional for graph entropies and for RT
entropies in a common static realization.  Our convention is
\begin{equation}
 I(X:Y\mid Z)=S(XZ)+S(YZ)-S(Z)-S(XYZ).
 \label{eq:cmi-convention}
\end{equation}

\paragraph{The five-party inequality.}

The five-party holographic entropy cone, the set of entropy vectors realizable
by five disjoint holographic boundary parties and their purifier, contains the
inequality
\cite{Bao:2015bfa,HernandezCuenca:2019wgh}
\begin{equation}
 \boxed{
 \QFive(0;a,b,c,d)
 =I(c:d\mid0)+I(a:b\mid0c)+I(a:b\mid0d)-I(a:b\mid0)
 \geq0.}
 \label{eq:Q5}
\end{equation}
Equivalently,
\begin{align}
 \QFive={}&S(0ab)+S(0ac)+S(0ad)+S(0bc)+S(0bd)
 \notag\\
 &-S(0a)-S(0b)-S(0cd)-S(0abc)-S(0abd).
 \label{eq:Q5entropy}
\end{align}
Algebraically, \eqref{eq:Q5} is the Ingleton combination with every term
conditioned on $0$.  The unconditioned Ingleton relation is a four-variable
linear-rank inequality familiar from matroid and network information theory;
conditional variants under classical conditional-independence hypotheses are
well known \cite{Studeny:2020ci}.  Our use of \eqref{eq:Q5} concerns its
radial light-cone specialization.

The four active outer parties $a,b,c,d$ must be disjoint boundary atoms.  A
common-atom realization places them in angular slots separated by nonzero
gaps, or buffers, so that their supports are disjoint, as required by
\eqref{eq:Q5-disjoint-support}.  The
distinguished party $0$, however, may be a composite union and may contain a
radial petal which has already been added to the central region.  This
asymmetry between the central and outer labels is the mechanism behind the
shell construction.

\paragraph{Finite shells and conditional monotonicity.}

Let $Z,B,C,D$ be fixed disjoint regions on the regulated slice.  Let $A_R$ be
a nested radial petal, and define the shell
\begin{equation}
 A_R\subset A_{R'},
 \qquad
 E_{R,R'}:=A_{R'}\setminus A_R,
 \qquad R'>R.
 \label{eq:shell}
\end{equation}
We restrict to radii for which $A_{R'}$ is disjoint from
$Z\cup B\cup C\cup D$.  Thus $A_R,E_{R,R'},Z,B,C,D$ are disjoint atoms,
apart from the radial adjacency of $A_R$ and its shell.
The shell occupies the same angular slot as $A_R$ but begins immediately
outside it; $B,C,D$ occupy three other buffered slots.  After $A_R$ is
absorbed into the center, the active outer labels $E_{R,R'},B,C,D$ therefore
have disjoint supports and satisfy the common-atom criterion.

For later convenience, define
\begin{equation}
 \mathcal I_X(R):=I(A_R:B\mid X),
 \qquad
 \mathfrak m(R):=I(C:D\mid ZA_R).
 \label{eq:IX-Markov-gap}
\end{equation}

The condition $\mathfrak m(R)=0$ is a Markov equality: it saturates SSA, or
equivalently sets the conditional mutual information to zero.  The locus on
which this equality holds is called the Markov face of the entropy cone.

\begin{theorem}[Exact finite shell bound]
\label{thm:finite-shell}
Suppose $A_{R'}$ is disjoint from $Z\cup B\cup C\cup D$.  Then, for every
$R'>R$,
\begin{equation}
 \boxed{
 0\leq \mathcal I_Z(R')-\mathcal I_Z(R)
 \leq \mathfrak m(R)
 +\mathcal I_{ZC}(R')-\mathcal I_{ZC}(R)
 +\mathcal I_{ZD}(R')-\mathcal I_{ZD}(R).}
 \label{eq:finite-shell-general}
\end{equation}
This statement is finite and requires no differentiability or Markov
assumption.
\end{theorem}

\begin{proof}
Apply \eqref{eq:Q5} with
\begin{equation}
 0'=ZA_R,
 \qquad
 (a,b,c,d)=(E_{R,R'},B,C,D).
 \label{eq:Q5-shell-assignment}
\end{equation}
Rearranging the resulting inequality gives
\begin{align}
 I(E_{R,R'}:B\mid ZA_R)
 \leq{}&I(C:D\mid ZA_R)
 +I(E_{R,R'}:B\mid ZA_RC)
 \notag\\
 &+I(E_{R,R'}:B\mid ZA_RD).
 \label{eq:Q5shell}
\end{align}
For $X=Z,ZC,ZD$, the conditional-mutual-information chain rule gives
\begin{equation}
 I(E_{R,R'}:B\mid XA_R)
 =\mathcal I_X(R')-\mathcal I_X(R).
 \label{eq:shell-chain}
\end{equation}
Substitution proves the upper bound.  The lower bound is the nonnegativity of
the left-hand side of \eqref{eq:shell-chain}, and therefore follows from SSA.
\end{proof}

The theorem displays the main difficulty in differentiating a generic
holographic inequality.  Set $R'=e^hR$.  If $\mathfrak m(R)>0$ remains fixed
as $h\to0^+$, the term $\mathfrak m(R)/h$ diverges.  Positivity of the finite
five-party gap is not enough; a local upper bound appears only when its
\emph{zeroth-order floor}, the term which remains as the shell thickness
vanishes, is zero.

\begin{corollary}[Conditional radial speed limit]
\label{cor:radial-speed}
Suppose that
\begin{equation}
 I(C:D\mid ZA_R)=0,
 \label{eq:one-Markov-condition}
\end{equation}
at the base radius $R$.  If the indicated right logarithmic derivatives
exist, then
\begin{equation}
 \boxed{
 0\leq R\partial_R^+ I(A_R:B\mid Z)
 \leq R\partial_R^+\!\left[
 I(A_R:B\mid ZC)+I(A_R:B\mid ZD)\right].}
 \label{eq:radial-speed}
\end{equation}
Only the Markov equality at the radius where the derivative is evaluated is
required.  Here $\partial_R^+$ denotes the derivative approached from larger
radii, and $R\partial_R^+$ is the corresponding one-sided derivative with
respect to $\log R$.
\end{corollary}

At either conformal endpoint, the common-null-sheet construction
automatically obeys \eqref{eq:one-Markov-condition}.  Indeed, set
$X=ZA_RC$ and $Y=ZA_RD$.  Since the atoms are disjoint,
$X\cap Y=ZA_R$ and $X\cup Y=ZA_RCD$, so $I(C:D\mid ZA_R)$ is precisely
the SSA gap for $X$ and $Y$.  The vacuum of a CFT is a quantum Markov
state for regions defined by cuts of a common null plane or light cone,
and this gap therefore vanishes
\cite{Casini:2017vbe,Casini:2017roe}.  Thus no independent Markov
hypothesis is needed for the endpoint CFT vacua after the controlled
null limit is taken.  Its persistence at intermediate scales in the
null-cone RG-flow setting, however, remains an additional physical
hypothesis.

To establish Corollary \ref{cor:radial-speed}, divide \eqref{eq:finite-shell-general} by
$h=\log(R'/R)$, impose \eqref{eq:one-Markov-condition}, and take
$h\to0^+$.  Because $B,Z,C,D$ remain fixed, the same statement can be
written directly in entropy variables as
\begin{align}
 0\leq{}&R\partial_R^+\!\left[S(A_RZ)-S(A_RBZ)\right]
 \notag\\
 \leq{}&R\partial_R^+\!\left[
 S(A_RZC)-S(A_RBZC)
 +S(A_RZD)-S(A_RBZD)\right].
 \label{eq:entropy-derivative}
\end{align}
This form makes the analogy with the Casini-Huerta construction
transparent \cite{Casini:2012ei}.  Its content is different: it bounds the
growth of one conditional correlation channel by two conditioned channels,
rather than defining a universal scalar $F$-function.

If \eqref{eq:one-Markov-condition} holds throughout an interval, define
\begin{equation}
 \PhiFive(R):=
 I(A_R:B\mid ZC)+I(A_R:B\mid ZD)-I(A_R:B\mid Z).
 \label{eq:Phi5}
\end{equation}
Equation~\eqref{eq:radial-speed} then implies
\begin{equation}
 R\partial_R^+\PhiFive(R)\geq0.
 \label{eq:Phi5mono}
\end{equation}
Thus $\PhiFive$ is nondecreasing as the nested petal radius increases.  We
call it a conditional nested-region monotone.  Its spectator regions
remain fixed, its endpoint values are not universal, and its all-interval
Markov condition is an additional physical hypothesis.  It should not be
identified with a candidate central charge.

\paragraph{Several moving petals.}

The same argument controls the common outward growth of two active regions.
Let $A_R$ and $B_R$ be nested families which remain mutually disjoint and
disjoint from $Z\cup C\cup D$ throughout the interval, and suppose that
\begin{equation}
 I(C:D\mid ZA_R)=0,
 \qquad
 I(C:D\mid ZB_R)=0.
 \label{eq:two-petal-Markov}
\end{equation}
Growing the two petals sequentially and applying
Theorem~\ref{thm:finite-shell} at each step gives
\begin{equation}
 \boxed{
 0\leq R\partial_R^+I(A_R:B_R\mid Z)
 \leq R\partial_R^+\!\left[
 I(A_R:B_R\mid ZC)+I(A_R:B_R\mid ZD)\right].}
 \label{eq:common-dilation}
\end{equation}
This is a common dilation of the two active petals; the spectators remain
fixed, so it is not a self-similar RG dilation of the complete configuration.
The two equalities in \eqref{eq:two-petal-Markov} are conditions on the
regions present before the infinitesimal shell is added; neither contains that
new shell.

More generally, let $A$ vary in a multiparameter scale-and-shape family while
$B,Z,C,D$ remain fixed.  Every tangent vector $v$ generated by an outward
nested deformation, which enlarges $A$ without removing any of its points,
obeys
\begin{equation}
 0\leq v\!\cdot\!\nabla I(A:B\mid Z)
 \leq v\!\cdot\!\nabla\!\left[
 I(A:B\mid ZC)+I(A:B\mid ZD)\right],
 \label{eq:mixed-tangent}
\end{equation}
when $I(C:D\mid ZA)=0$ at the base point.  If fixed absolute spectator
positions are described by dimensionless variables $u_i=\ell_i/R$, their
fixed-position direction is
\begin{equation}
 v\!\cdot\!\nabla=R\partial_R-\sum_i u_i\partial_{u_i}.
 \label{eq:mixed-vector}
\end{equation}
Equation~\eqref{eq:mixed-tangent} applies only where this direction remains
outward nested.  Moving $Z,C,D$, or moving several active petals without a
sequential shell decomposition, introduces additional terms and Markov
requirements.

\paragraph{Independence and sharpness.}

Only the lower inequality in \eqref{eq:radial-speed} follows from SSA.  We
now show that its upper inequality remains genuinely five-party in the
infinitesimal-shell limit $R'\to R^+$.  We do so by exhibiting abstract
entropy-like set functions which obey SSA and MMI but violate the upper bound;
we do not claim that they arise in a Lorentz-invariant QFT.

Let the nonpurifying labels be $\mathcal N=\{0,e,b,c,d\}$.  Define $H_0$
by declaring $e$ to be a loop, meaning a zero-rank element whose addition does
not change the set function, $H_0(Ue)=H_0(U)$, and assigning, on
$\{0,b,c,d\}$,
\begin{equation}
 H_0(U)=
 \begin{cases}
 0,&|U|=0,\\
 2,&|U|=1,\\
 4,&|U|=2,\\
 6,&U=\{0,c,d\},\\
 5,&|U|=3,\ U\neq\{0,c,d\},\\
 6,&|U|=4.
 \end{cases}.
 \label{eq:H0}
\end{equation}
Define $H_1$ on all five labels by
\begin{equation}
 H_1(U)=
 \begin{cases}
 0,&|U|=0,\\
 2,&|U|=1,\\
 4,&|U|=2,\\
 6,&U\in\{\{0,c,d\},\{e,c,d\}\},\\
 5,&|U|=3\text{ otherwise},\\
 6,&|U|\geq4.
 \end{cases}.
 \label{eq:H1}
\end{equation}
After adjoining a purifier $\mathsf p$, extend these functions by purity:
\begin{equation}
 \widetilde H_i(U)=
 \begin{cases}
 H_i(U),&\mathsf p\notin U,\\
 H_i\!\left(\mathcal N\setminus(U\setminus\{\mathsf p\})\right),
 &\mathsf p\in U.
 \end{cases}.
 \label{eq:pure-extension}
\end{equation}

A purifier image of an inequality is obtained by exchanging an ordinary label
with $\mathsf p$ and then using purity, $S(U)=S(U^c)$.  A finite enumeration
over disjoint composite systems $X,Y,W$, with $X,Y$
nonempty and $W$ arbitrary, gives nonnegative conditional mutual information.
For disjoint nonempty $X,Y,W$, define the tripartite information
\begin{equation}
 I_3(X:Y:W):=S(X)+S(Y)+S(W)-S(XY)-S(XW)-S(YW)+S(XYW).
 \label{eq:tripartite-information-definition}
\end{equation}
MMI is the condition $I_3\leq0$.  The same enumeration gives
\begin{equation}
 -I_3^{\widetilde H_0}(X:Y:W)\in\{0,1\},
 \qquad
 -I_3^{\widetilde H_1}(X:Y:W)\in\{0,1,2\}.
 \label{eq:mmi-enumeration}
\end{equation}
Thus both functions obey SSA and MMI, together with all purifier images of
these inequalities.  The segment
\begin{equation}
 H_s=(1-s)H_0+sH_1,
 \qquad 0\leq s\leq1,
 \label{eq:Hs}
\end{equation}
stays inside the polyhedral cone cut out by the linear SSA and MMI
inequalities.  Nevertheless,
\begin{align}
 I_{H_s}(c:d\mid0)&=0,
 \notag\\
 \left.\partial_sI_{H_s}(e:b\mid0)\right|_{s=0}&=1,
 &
 \left.\partial_sI_{H_s}(e:b\mid0c)\right|_{s=0}&=0,
 &
 \left.\partial_sI_{H_s}(e:b\mid0d)\right|_{s=0}&=0.
 \label{eq:separator-rates}
\end{align}
The tangent violates the upper shell-rate bound while preserving its Markov
hypothesis.  We conclude that this upper bound cannot be obtained as a
nonnegative linear combination of SSA, MMI, or their purifier images.  This
counterexample establishes that the information inequality is irredundant
within that system; it does not supply a physical radial entropy trajectory.

\paragraph{A sharp graph chamber.}

A complementary graph model gives an open region of edge-capacity space whose
entropy vectors lie on the required Markov face.  We call an open
region of edge-capacity space on which the same cuts remain minimizing a
\emph{minimum-cut chamber}; when its entropy vectors lie on the Markov face,
we call it a Markov chamber.  Consider bulk vertices $u,v,w$
with internal edge capacities
\begin{equation}
 w_{uv}=3,
 \qquad w_{uw}=5,
 \qquad w_{vw}=3.
 \label{eq:bulkedges}
\end{equation}
Attach fixed boundary terminals by
\begin{equation}
 Z-v:6,
 \qquad B-w:5,
 \qquad C-u:4,
 \qquad D-w:1,
 \qquad \mathsf p-w:4,
 \label{eq:boundaryedges}
\end{equation}
where $\mathsf p$ is the purifier.  Finally, attach to $u$ many boundary
terminals whose small edge capacities discretize radial increments; we call
them \emph{radial microterminals}.  Their fixed collection, of total capacity
$L\leq2$, is the reservoir.  At a chosen scale, partition this reservoir into
the already-included region $A$, the next shell $E$, and the future
complement, with capacities
\begin{equation}
 a,e,r\geq0,
 \qquad a+e+r=L.
 \label{eq:reservoir}
\end{equation}
The graph and its edge weights remain fixed; changing $R$ changes only which
reservoir microterminals belong to $A_R$.

\begin{proposition}[Sharp graph chamber on the Markov face]
\label{prop:sharp-graph}
For $0\leq L\leq2$, direct minimum-cut enumeration gives
\begin{equation}
 \boxed{
 \begin{aligned}
 I(C:D\mid ZA)&=0,
 &\qquad I(E:B\mid ZA)&=2e,\\
 I(E:B\mid ZAC)&=0,
 &\qquad I(E:B\mid ZAD)&=2e.
 \end{aligned}}
 \label{eq:graph-CMIs}
\end{equation}
The finite shell theorem is therefore saturated nontrivially throughout the
open minimum-cut chamber $a,e,r>0$ with $L<2$ on the Markov face.
\end{proposition}

The relevant minimum cuts are listed in
Appendix~\ref{app:q5-min-cuts}.  Writing the cumulative included reservoir
capacity as $a(R)$, the same chamber has
\begin{equation}
 I(A_R:B\mid Z)=2a(R),
 \qquad
 I(A_R:B\mid ZC)=0,
 \qquad
 I(A_R:B\mid ZD)=2a(R).
 \label{eq:graph-M}
\end{equation}
For a finite reservoir, $a(R)$ is a step function.  Thus the proposition
establishes sharp finite-shell saturation, including shells of nonzero
capacity, but a single finite graph does not exhibit a nonzero smooth local
rate.  A refining sequence of reservoirs, in which the individual
microterminal capacities tend to zero, can instead converge to a continuous
$a(R)$.  If that limit is differentiable and stays in the same cut chamber,
\begin{equation}
 R\partial_R I(A_R:B\mid Z)=2R a'(R),
 \label{eq:graph-rate}
\end{equation}
and the upper speed limit is saturated wherever the limiting derivative
exists.  The nonzero differential statement therefore belongs to the
continuum-reservoir limit, while every member of the refining sequence is an
exact finite holographic entropy vector.

Every finite reservoir discretization is an exact holographic graph model.
The graph-to-hyperbolic-geometry construction therefore supplies a
time-symmetric RT realization \cite{Bao:2015bfa}.  In this construction,
\emph{thickening} replaces weighted graph edges by narrow hyperbolic bridges
whose bottleneck sizes reproduce the edge capacities.  The resulting RT
geometry generally has several asymptotic boundaries, however, and does not
prove that the same
chamber occurs in a smooth, homogeneous, single-boundary RG-flow vacuum.
The legal one-boundary light-cone placement and the sharp graph realization
are presently separate consistency checks.

The shell theorem succeeds by using exactly the loopholes left open by the
earlier obstructions.  Its scale-dependent central party contains $A_R$, and
its zeroth-order gap vanishes only on the selected Markov trajectory.  The
result controls nested-petal radial growth.  With the spectator positions held
fixed, it is a mixed scale/shape response rather than a common self-similar
dilation of the entire configuration.  Its conditionality remains a
substantial physical restriction.  We next remove the Markov hypothesis by
taking a continuum limit of an infinite family of angular inequalities.

\subsection{Odd-cyclic constraints from angular refinement}
\label{sec:angular}

The moving-shell construction gives a local correlation-growth bound, but it
pays for this result with a Markov equality.  We now follow a different route
and take a continuum limit of an infinite family of finite-party
inequalities.  The resulting constraint requires no entropy equality, but it
controls an angular shape derivative rather than a radial derivative by
itself.

Let $n=2k+1$ and excise a circle of radius $r_0$ about the cone tip.  Partition
the resulting regulated annulus into equal angular sectors
$A_1,\ldots,A_n$.  The odd-cyclic inequalities compare cyclic sums of
entropies of consecutive blocks.  For $n=2k+1$, the symmetric maximal-block
member compares every block of $k+1$ adjacent sectors with the corresponding
block of $k$ sectors:
\begin{align}
 C_n={}&\sum_{i=1}^{n}S(A_iA_{i+1}\cdots A_{i+k})
 -\sum_{i=1}^{n}S(A_iA_{i+1}\cdots A_{i+k-1})
 \notag\\
 &-S(A_1A_2\cdots A_n)\geq0,
 \label{eq:cyclic}
\end{align}
where the indices are understood cyclically
\cite{Bao:2015bfa,HernandezCuenca:2019wgh}.  Let
$S_W^{(r_0)}(R,\theta)$ denote the entropy of the regulated annular wedge of
opening angle $\theta$, and let $S_A(R,r_0)$ denote the entropy of the full
annulus.  Rotation invariance turns \eqref{eq:cyclic} into the exact finite
inequality
\begin{equation}
 n\left[
 S_W^{(r_0)}\!\left(R,\pi+\frac{\pi}{n}\right)
 -S_W^{(r_0)}\!\left(R,\pi-\frac{\pi}{n}\right)
 \right]\geq S_A(R,r_0).
 \label{eq:finite-angular}
\end{equation}
This formula is already a constraint at every RG scale.  Its continuum limit
extracts a first shape derivative.

\begin{theorem}[Fixed-scale angular inequality]
\label{thm:angular-continuum}
Assume that the finite-$n$ odd-cyclic inequality is valid for each commonly
regulated entropy vector.  Under the regulated null and uniform continuum
assumptions of Section~\ref{sec:ssa-baseline}, the $n\to\infty$ limit at
fixed nonzero tip regulator $r_0$ is
\begin{equation}
 2\pi\left.\partial_\theta S_W^{(r_0)}(R,\theta)
 \right|_{\theta=\pi}-S_A(R,r_0)\geq0.
 \label{eq:angular-annulus}
\end{equation}
Suppose in addition that the joint $r_0\to0$ and UV-regulator limit of this
complete balanced gap exists and agrees with the corresponding
puncture-filled disk combination,
\begin{equation}
 \lim_{r_0\to0}\left[
 2\pi\left.\partial_\theta S_W^{(r_0)}(R,\theta)
 \right|_{\theta=\pi}-S_A(R,r_0)\right]
 =2\pi\left.\partial_\theta S_W(R,\theta)\right|_{\theta=\pi}-S_D(R).
 \label{eq:angular-tip-limit}
\end{equation}
Here the UV regulator is understood to be removed in the common correlated
prescription of Section~\ref{sec:ssa-baseline}.  Then
\begin{equation}
 \boxed{
 \mathcal G_{\rm ang}(R):=
 2\pi\left.\partial_\theta S_W(R,\theta)\right|_{\theta=\pi}
 -S_D(R)\geq0.}
 \label{eq:angular-theorem}
\end{equation}
No Markov saturation is required.
\end{theorem}

The finite-$n$ premise is automatic for graph entropies and for RT entropies
with a common static realization.  Its use for the regulated covariant HRT
regions is the standing assumption stated in the Introduction.

\begin{proof}
At fixed $r_0$, divide the angular difference in
\eqref{eq:finite-angular} by $2\pi/n$.  Uniform convergence of the regulated
symmetric difference quotient gives \eqref{eq:angular-annulus}.  Positivity is
preserved in this joint tip and UV-regulator limit; using
\eqref{eq:angular-tip-limit} gives \eqref{eq:angular-theorem}.
\end{proof}

The order of limits is important.  The complete finite gap is formed before
taking the null, smoothing, $n\to\infty$, tip, and UV limits.  On a
tip-regulated annulus, the leading perimeter terms cancel between the two
wedge sums and the full-annulus term.  That cancellation alone does not
justify filling the puncture: the convergence and identification in
\eqref{eq:angular-tip-limit} are additional assumptions about the
continuously chosen entropy branch, meaning the RT/HRT saddle which remains
minimizing, rather than a formal term-by-term operation.

The sign in \eqref{eq:angular-theorem} is not a consequence of SSA.  For each
finite $n$, consider the rotation-symmetric set function
\begin{equation}
 H_n(U)=\sqrt{\frac{|U|}{n}}.
 \label{eq:Hn}
\end{equation}
It depends only on the cardinality $|U|$ and is a polymatroid, meaning a
normalized, monotone, submodular set function.  Submodularity is the abstract
set-function form of SSA, so $H_n$ obeys every SSA inequality.  However,
\begin{align}
 C_n[H_n]
 &=n\left[\sqrt{\frac{k+1}{n}}-\sqrt{\frac{k}{n}}\right]-1
 \notag\\
 &=\frac{\sqrt{2n}}{\sqrt{n+1}+\sqrt{n-1}}-1
 \longrightarrow\frac1{\sqrt2}-1<0.
 \label{eq:Hn-violation}
\end{align}
Thus the continuum inequality retains information beyond SSA even after
arbitrary angular refinement.  As in the finite separator constructed above,
this abstract polymatroid need not be the
entropy vector of a Lorentz-invariant QFT.

\paragraph{Lorentz Ward identity and radial/shape mixing.}

Although \eqref{eq:angular-theorem} is a shape constraint, Lorentz covariance
relates its angular derivative to a deformation containing the physical scale
$R$.  For an outer-cut profile $f(\phi)$, define the first variation at fixed
angular endpoints by
\begin{equation}
 \delta_fS_W(R,\theta)
 :=\left.\frac{\dd}{\dd\epsilon}
 S\!\left[\rho_{\rm out}(\phi)=R(1+\epsilon f(\phi))\right]
 \right|_{\epsilon=0}.
 \label{eq:shape-variation}
\end{equation}
An infinitesimal boost along the wedge bisector acts, to first order, as
\begin{equation}
 \rho'=\rho(1+\epsilon\cos\phi),
 \qquad
 \phi'=\phi-\epsilon\sin\phi.
 \label{eq:boost}
\end{equation}
The entropy is invariant under this Lorentz transformation.  Accounting for
the motion of the two angular endpoints gives the Ward identity
\begin{equation}
 \delta_{\cos\phi}S_W(R,\theta)
 =2\sin(\theta/2)\,\partial_\theta S_W(R,\theta).
 \label{eq:ward}
\end{equation}
At $\theta=\pi$, linearity of first variations and
$\delta_1S_W=R\partial_RS_W$ rewrite
\eqref{eq:angular-theorem}, for any constant $c$, as
\begin{equation}
 \pi c\,R\partial_RS_W(R,\pi)
 +\pi\delta_{\cos\phi-c}S_W(R,\pi)
 \geq S_D(R).
 \label{eq:mixed-c}
\end{equation}
The choice $c=2/\pi$ makes the residual angular profile have zero mean and
gives
\begin{equation}
 \boxed{
 2R\partial_RS_W(R,\pi)
 +\pi\delta_{\cos\phi-2/\pi}S_W(R,\pi)
 \geq S_D(R).}
 \label{eq:mixed-main}
\end{equation}

Equation~\eqref{eq:mixed-main} is a covariant rewriting of the angular
inequality, not a second positivity theorem.  In particular, the zero-mean
shape response has no definite sign.  Discarding it would therefore not
produce a valid radial monotone.

At a finite tip radius $r_0$, the boost moves both radial cuts.  The
regulator-safe Ward identity replaces
\begin{equation}
 R\partial_R\longrightarrow R\partial_R+r_0\partial_{r_0},
 \label{eq:tip-derivative}
\end{equation}
and deforms both cuts by the residual profile.  The complete balanced
combination must be formed before taking $r_0\to0$; its separate terms need
not be individually scheme independent.

The mixed form remains beyond SSA even if one imposes the abstract Ward
identity.  For a nonnegative null-cut profile $\gamma(\phi)$, define
\begin{equation}
 H(A_\gamma)=\sqrt{\int_A\gamma(\phi)\,\dd\phi}.
 \label{eq:lorentz-witness}
\end{equation}
This functional is submodular and invariant because
$\gamma'\dd\phi'=\gamma\dd\phi$ under a tip-fixing Lorentz transformation.
For round cuts,
$H_W(R,\theta)=\sqrt{R\theta}$ and
$H_D(R)=\sqrt{2\pi R}$.  Its zero-mean residual response vanishes, whereas
\begin{equation}
 2R\partial_RH_W(R,\pi)=\sqrt{\pi R}
 <\sqrt{2\pi R}=H_D(R).
 \label{eq:lorentz-witness-violation}
\end{equation}
It therefore violates \eqref{eq:mixed-main} while obeying SSA and the Ward
identity.

The odd-cyclic construction evades the finite-party pair projection through
the correlated limit $n\to\infty$.  What survives is a genuine beyond-SSA
shape inequality at each fixed scale.  Its Lorentz form contains a radial
derivative, but the accompanying shape term marks the remaining obstruction
to a scalar RG monotone.

\section{Conclusions and outlook}
\label{sec:conclusion}

In this work, we have asked whether holographic entropy inequalities beyond SSA can play the role that SSA plays in the Casini-Huerta proof.  The answer is affirmative, but qualified.  These inequalities do constrain entanglement response along holographic RG flows, but our analysis does not produce a second universal radial monotone.

We first characterized the geometric and infinitesimal obstructions.  A
collection of additive null cuts can be realized by one set of disjoint
regions only when its active subsets form a chain on almost every generator.
For the disjoint-support bump family, every nonzero centered balanced
inequality which is nonnegative on the holographic graph cone has a nontrivial
pair projection, and its common-amplitude quadratic response is a nonnegative
sum of susceptibilities already controlled by SSA.  The six-party functional $G_6$
demonstrates that this projection can discard genuinely non-Shannon
information.  Exact deletion and finite-support scale constructions do not
restore that information under the assumptions we have stated.  The
domain-wall SCT calculation gives an independent bulk manifestation of the
same pairwise structure, although its deformation parameter $\kappa$ is not
the RG scale $R$.

We then obtained two positive results which evade the obstruction in
different ways.  For the five-party inequality, absorbing the already-grown
petal into the conditioning system isolates one zeroth-order conditional mutual
information.  When this quantity vanishes, the finite shell inequality gives
a local upper speed limit on nested-petal correlation growth.  The upper bound is
not a consequence of SSA or MMI as an entropy-cone statement.  Finite
holographic graphs sharply saturate the finite-shell bound, while a refining
continuum-reservoir limit saturates the differential bound.  If the Markov
condition holds throughout an interval, the same bound defines a conditional monotone.  We
have not shown that the required Markov chamber occurs in a smooth
single-boundary holographic RG-flow vacuum.

The odd-cyclic construction gives a complementary result.  It requires no
Markov equality and yields an angular differential inequality at every scale,
under the stated finite-$n$ HRT-validity, common-regulator,
uniform-continuum, and correlated tip/UV-limit assumptions.  Lorentz
covariance exposes an RG-scale derivative in an equivalent mixed response.
However, the remaining zero-mean shape variation is not sign-definite.  We
therefore cannot extract a universal radial monotone from this inequality
alone.

Taken together, these statements show that higher holographic entropy
inequalities can constrain entanglement response at points along an RG flow,
but they do not reproduce all ingredients of the SSA construction in a single
observable.  The five-party result controls nested-petal radial growth but
requires a Markov equality, whereas the odd-cyclic result requires no such
equality but retains a sign-indefinite
shape term.  A universal scalar beyond-SSA monotone therefore remains open.

There are several concrete directions in which the present analysis can be
extended.

\begin{enumerate}[label=(\roman*),leftmargin=2.1em]
 \item \emph{Physical realization of the Markov chamber.}
 The sharp graph model should be replaced by a smooth, homogeneous,
 single-boundary bulk geometry.  One possible route is to search directly in
 asymptotically AdS domain walls for RT or HRT phases satisfying
 $I(C:D\mid ZA_R)=0$ while the shell response remains nonzero.  Such a
 construction would determine whether the conditional monotone can occur on
 a genuine RG trajectory and whether this vanishing-CMI condition can persist
 over a finite interval.

 \item \emph{Separating cyclic scale and shape responses.}
 The obstruction in the odd-cyclic theorem is the sign-indefinite zero-mean
 shape derivative.  Averaging over boost directions, decomposing the response
 into angular harmonics, or combining the cyclic inequalities with
 independent shape-variation bounds may control this term.  A successful
 construction would turn the Markov-independent angular theorem into a radial
 statement.

 \item \emph{Higher-order and nonlocal probes.}
 The quadratic response is forced onto pair data, but our argument does not
 constrain cubic or higher orders.  Determining whether any higher response
 of $G_6$ retains its non-Shannon coefficients would test where genuinely
 multipartite local information first appears.  Similarly,
 continuum kernels involving infinitely many scales may evade
 Theorem~\ref{thm:scale-no-go-short}, whose proof assumes a Laurent polynomial
 with finitely many monomials, provided their convergence and endpoint
 behavior can be controlled.

 \item \emph{Covariant and regulator control.}
 A fully intrinsic derivation should establish the required holographic
 inequalities and all null, angular, and shell limits directly for HRT
 surfaces.  In the bulk SCT calculation this includes the outstanding compact
 cap estimate.  More generally, a regulator-independent treatment would
 clarify which separated terms in the mixed radial/shape inequality define
 observables individually, rather than only through their balanced
 combination.
\end{enumerate}
These directions isolate the ingredients that a genuine beyond-SSA analogue of the $F$-theorem would  require: an attainable physical limit or entropy-saturation condition that removes the finite entropy gap, a sign  definite response along a self-similar RG trajectory, and fixed-point data with a universal interpretation.  Establishing all three ingredients in one construction remains open.

\section*{Acknowledgments}

We are grateful to Yikun Jiang, Jacob March, and Joydeep Naskar for productive conversations about this work, and especially to Keiichiro Furuya for many useful discussions and collaboration in the initial stages of this project. We also acknowledge the language model ChatGPT-5.6 Sol, which was used in this research. The authors take full responsibility for the accuracy of this work, having reviewed, verified, and approved all AI-generated content. N.\,B. is supported by Northeastern University Department of Physics, Brookhaven National Laboratory, and the U.S. Department of Energy ASCR EXPRESS grant, Novel Quantum Algorithms from Fast Classical Transforms. C.\,F. is supported by the National Science Foundation under Cooperative Agreement PHY-2019786 (the NSF AI Institute for Artificial Intelligence and Fundamental Interactions).

\appendix

\section{Auxiliary graph constructions}
\label{app:selectors}

This appendix collects the graph details behind both the deletion
classifications and the sharp five-party shell example.  For deletion, the
first graph family realizes every upper-set constraint in
\eqref{eq:upper-set-short}, while the second isolates one conditional mutual
information.  The minimum cuts for the five-party example are listed in
Subsection~\ref{app:q5-min-cuts}.

Let $\mathcal F\subseteq2^C$ be an upper set and let $\mathcal M$ be its
inclusion-minimal members.  They form an antichain, meaning that no member
contains another.  Membership is the monotone Boolean function
\begin{equation}
 \mathbf1_{\mathcal F}(U)
 =\bigvee_{K\in\mathcal M}\bigwedge_{i\in K}\mathbf1_{i\in U},
 \label{eq:selector-DNF-app}
\end{equation}
where $\wedge$ and $\vee$ denote Boolean AND and OR, respectively.  Thus
$U\in\mathcal F$ precisely when it contains every element of at least one
$K\in\mathcal M$.
Call the cut side containing \(0\) the \(1\)-side and the purifier side the
\(0\)-side.  For the region \(0U\), terminal \(i\in C\) has value
\(x_i=\mathbf1_{i\in U}\).

For each \(K\in\mathcal M\) of size \(q\geq2\), introduce a vertex \(a_K\).
Connect it to each input in \(K\) with weight \(4\) and to the purifier with
weight \(4(q-1)\).  If \(k\) inputs lie on the \(1\)-side, its two costs are
\begin{equation}
 E_K(0)=4k,
 \qquad
 E_K(1)=4(q-k)+4(q-1).
 \label{eq:and-cost-app}
\end{equation}
The \(1\)-side wins by \(4\) when all inputs are one; otherwise the
\(0\)-side wins by at least \(4\).  Thus the minimum cut uniquely sets
\(a_K=\bigwedge_{i\in K}x_i\).  A singleton \(K\) uses its terminal directly.

Denote these outputs by \(y_1,\ldots,y_m\).  Join an output vertex \(v\) to
each \(y_\ell\) with weight \(2\), and for \(m\geq2\) join \(v\) to the center
with weight \(2(m-1)\).  If \(k\) signals equal one, then
\begin{equation}
 E_v(0)=2k+2(m-1),
 \qquad
 E_v(1)=2(m-k).
 \label{eq:or-cost-app}
\end{equation}
Hence \(v=1\) exactly when at least one signal is one, with a unique
minimizing side.

Attach the probe terminal \(p\) to \(v\) with a unit edge.  The probe cannot
change the gate values, namely the Boolean values at the vertices $a_K$ and
$v$: correcting an AND output saves at least \(4\) and
changes its edge to \(v\) by at most \(2\); correcting \(v\) then saves at
least \(2\) and changes the probe cost by at most \(1\).  Attaching $v$ to
the purifier or center with weight \(2\) fixes the required constant output
and gives the empty or full upper set.

Let \(B_{\mathcal F}(U)\) be the cut cost without the probe edge.  The probe
lies on the purifier side for \(S(0U)\) and on the center side for
\(S(0Up)\), giving
\begin{align}
 S(0U)&=B_{\mathcal F}(U)+\mathbf1_{\mathcal F}(U),
 \notag\\
 S(0Up)&=B_{\mathcal F}(U)+1-\mathbf1_{\mathcal F}(U).
 \label{eq:selector-costs-app}
\end{align}
Their difference is
\begin{equation}
 S(0Up)-S(0U)=1-2\mathbf1_{\mathcal F}(U),
 \label{eq:selector-signal-app}
\end{equation}
which is \eqref{eq:upper-selector-short}.

For the converse in Theorem~\ref{thm:one-deletion-short}, connect a source
to \(U\) with capacity \(b_U\) when \(b_U>0\), connect \(U\) to a sink with
capacity \(-b_U\) when \(b_U<0\), and give every upward Boolean-lattice edge
capacity larger than the total positive supply
$\sum_{U:\,b_U>0}b_U$.  A minimum cut cannot cross
a lattice edge, so its source-side vertices form an upper set
\(\mathcal F\).  Its capacity is the total supply minus
\(\sum_{U\in\mathcal F}b_U\), which is at least the total supply by
\eqref{eq:upper-set-short}.  Max-flow/min-cut therefore saturates all source
and sink edges, leaving a nonnegative upward flow of divergence \(b\).

For the double-deletion test, fix \(U_*\subseteq C\).  Assign positive star
weights \(a_i\) with distinct subset sums
\(a(V)=\sum_{i\in V}a_i\), and let \(\delta\) be the smallest separation
between two sums.  Give \(p,q\) weight \(\epsilon\), with
\(2\epsilon<\delta\), the center weight \(K\), and the purifier weight
\begin{equation}
 K+2a(U_*)-a(C),
 \label{eq:double-star-weights-app}
\end{equation}
where \(K\) makes both weights positive.  The midpoint of the total star
weight is \(K+a(U_*)+\epsilon\).  The quantity \(I(p:q\mid0V)\) samples the
cut function at
\begin{equation}
 K+a(V),\quad K+a(V)+\epsilon,\quad
 K+a(V)+\epsilon,\quad K+a(V)+2\epsilon.
 \label{eq:double-star-samples-app}
\end{equation}
These values form its discrete second difference and straddle the unique
kink only for \(V=U_*\).  Hence
\(I(p:q\mid0V)=2\epsilon\delta_{V,U_*}\), as in
\eqref{eq:isolated-CMI-short}.  The graph therefore isolates one
conditioning set.

\subsection{Minimum cuts for the sharp five-party graph}
\label{app:q5-min-cuts}

For completeness, we record the minimum cuts used to establish the graph
realization in Subsection~\ref{sec:radial}.  Write
$r=L-a-e$.  The required entropies are
\begin{center}
\renewcommand{\arraystretch}{1.12}
\begin{tabular}{@{}ll@{\qquad}ll@{}}
\toprule
Region & Entropy & Region & Entropy\\
\midrule
$ZA$       & $6+a$       & $ZAC$      & $8+e+r$\\
$ZAD$      & $7+a$       & $ZACD$     & $9+e+r$\\
$ZAE$      & $6+a+e$     & $ZAB$      & $9+e+r$\\
$ZAEB$     & $9+r$       & $ZACE$     & $8+r$\\
$ZACB$     & $5+e+r$     & $ZACEB$    & $5+r$\\
$ZADE$     & $7+a+e$     & $ZADB$     & $8+e+r$\\
$ZADEB$    & $8+r$       &             & \\
\bottomrule
\end{tabular}
\end{center}
Substituting these values into the four conditional mutual informations in
Eq.~\eqref{eq:graph-CMIs} gives the claimed result.  The same cuts remain
minimal throughout the interior $a,e,r>0$ with $L<2$.  The saturation is
therefore an open minimum-cut chamber rather than an isolated kink.

The logical dependence of the radial argument may be summarized as
\begin{equation}
 Q_5\geq0
 \quad\text{and}\quad
 I(C:D\mid ZA_R)=0
 \quad\Longrightarrow\quad
 \text{finite shell bound}
 \quad\Longrightarrow\quad
 \text{right radial speed limit}.
 \label{eq:q5-logical-chain}
\end{equation}
Independently, the finite odd-cyclic inequality gives the angular continuum
bound, which the Lorentz Ward identity rewrites as the mixed radial/shape
constraint.

\section{Special conformal transformation at finite cutoff}
\label{app:bulk-pair-response}

We derive Eq.~\eqref{eq:K-positive-representation} on the regulated
noncompact branch.  Both UV and remote regulators remain in place while the
embedding is expanded and the pair subtraction
$A_{\{i\}}+A_{\{j\}}-A_\varnothing-A_{\{i,j\}}$ is formed.  This order is
necessary because fields with zero Dirichlet data at the conformal boundary
can carry flux through a finite-cutoff surface.  After deriving the exact
identity with all boundary terms, we state sufficient conditions for
removing the regulators.  The compact surface requires a joint, or
\emph{diagonal}, limit in which the remote cutoff scales as
$\ell=|\kappa|^{-1}$, together with an estimate that remains unproved.
We use the convention
\begin{equation}
 U=x^+,
 \qquad V=x^-.
 \label{eq:app-UV-convention}
\end{equation}

\subsection{Finite-cutoff expansion}

Parameterize the codimension-two graph through order $\kappa^2$ by
$X^M(y,z)=(U,V,y,z)$.  Before any expansion in the SCT parameter, the
induced determinant is
\begin{equation}
 \det g_{\rm ind}
 =a^2\left[1-\nabla U\mathbin{\cdot}\nabla V
 -\frac14\left(
 |\nabla U|^2|\nabla V|^2
 -(\nabla U\mathbin{\cdot}\nabla V)^2\right)\right].
 \label{eq:app-induced-determinant}
\end{equation}
Setting $V=0$ and varying with respect to $V$ gives
$\mathcal DU=0$.  Hence the undeformed null-plane graph with boundary value
$f_T$ is $U=F_T$, where $F_T$ is the bulk extension defined in
Eq.~\eqref{eq:additive-cuts}.

The SCT boundary data require some care because the transformed graph is a
function of the final coordinate.  Start from
\begin{equation}
 X^+=\kappa Y^2+f_T(Y),
 \qquad X^-=0,
\end{equation}
and apply
\begin{equation}
 x^+=\frac{X^+}{1+\kappa X^+},
 \qquad
 x^-=-\frac{\kappa Y^2}{1+\kappa X^+},
 \qquad
 y=\frac{Y}{1+\kappa X^+}.
 \label{eq:app-null-sct}
\end{equation}
The expansion must hold the final coordinate $y$ fixed, not the initial
coordinate $Y$.  With this prescription, the boundary values are
\begin{align}
 U|_{z=0}
 &=f_T+\kappa\bigl(y^2+yf_Tf_T'-f_T^2\bigr)+O(\kappa^2),
 \notag\\
 V|_{z=0}
 &=-\kappa y^2-\kappa^2y^2f_T+O(\kappa^3).
 \label{eq:app-fixed-y-boundary-jet}
\end{align}
These boundary values require the embedding expansion
\begin{equation}
 U=F_T+\kappa U_{1,T}+O(\kappa^2),
 \qquad
 V=\kappa N+\kappa^2V_{2,T}+O(\kappa^3),
 \label{eq:app-graph-kappa-expansion}
\end{equation}
with
\begin{equation}
 N|_0=-y^2,
 \qquad
 U_{1,T}|_0=y^2+yf_Tf_T'-f_T^2,
 \qquad
 V_{2,T}|_0=-y^2f_T=N|_0f_T.
 \label{eq:app-higher-jacobi-data}
\end{equation}
The term $y^2$ in $U_{1,T}$ is independent of the subset label and cancels
from the four-term pair subtraction.  It cannot be omitted from the
embedding expansion, however, because it is part of the SCT boundary data.

\begin{samepage}
Expanding the Euler-Lagrange equations gives
\begin{align}
 \mathcal DN&=0,
 \notag\\
 \mathcal DU_{1,T}
 &=-\frac12\partial_\mu\!\left(
 a|\nabla F_T|^2\partial_\mu N\right),
 \notag\\
 \mathcal DV_{2,T}
 &=-\frac12\partial_\mu\!\left(
 a|\nabla N|^2\partial_\mu F_T\right).
 \label{eq:app-inhomogeneous-jacobi}
\end{align}
\end{samepage}
Here $\mu$ runs over $y,z$.  The last equation implies that $V_{2,T}$ is
affine in $F_T$.  Substitution into
Eq.~\eqref{eq:app-induced-determinant}, followed by expansion of the square
root through order $\kappa^2$, yields
\begin{align}
 A_{\rm reg}[T]=\int_\Omega a\bigg[1
 &-\frac\kappa2\nabla F_T\mathbin{\cdot}\nabla N
 -\frac{\kappa^2}{2}\left(
   \nabla U_{1,T}\mathbin{\cdot}\nabla N
  +\nabla F_T\mathbin{\cdot}\nabla V_{2,T}\right)
 \notag\\
 &-\frac{\kappa^2}{8}|\nabla F_T|^2|\nabla N|^2
 \bigg]+A_{\rm ct}+O(\kappa^3).
 \label{eq:app-area-second-order}
\end{align}
The domain $\Omega$ in this formula is still truncated at $z=\delta$ and at
a remote radial and transverse boundary.  All four areas in the pair
combination use the same counterterm prescription.  Thus the pair
subtraction can be carried out before either boundary is removed.

At finite cutoffs, zero boundary values do not imply zero boundary flux.
The pair subtraction removes the common profile and selects the terms
bilinear in the two bumps.  Define
\begin{align}
 \mathfrak m_{ij}A
 &:=A_{\{i\}}+A_{\{j\}}-A_\varnothing-A_{\{i,j\}},
 \qquad H_{ij}:=F_iF_j,
 \notag\\
 W_{ij}:=\mathfrak m_{ij}U_1
 &=U_{1,\{i\}}+U_{1,\{j\}}-U_{1,\varnothing}-U_{1,\{i,j\}},
\end{align}
and let $V_i,V_j$ be the increments of $V_{2,T}$ that are linear in $F_i,F_j$.
The pair field $W_{ij}$ has zero Dirichlet data at the conformal boundary:
\begin{equation}
 W_{ij}\big|_{z=0}
 =-y(f_if_j'+f_jf_i')+2f_if_j=0.
 \label{eq:app-W-zero-data}
\end{equation}
The last equality follows from the disjoint supports, but the normal
derivative of $W_{ij}$ can still contribute at finite cutoff.  Since
$\mathcal DN=0$, integration by parts gives
\begin{equation}
 -\frac12\int_\Omega a\,
 \nabla W_{ij}\mathbin{\cdot}\nabla N
 =-\frac12\int_{\partial\Omega}aW_{ij}\partial_nN.
 \label{eq:app-W-boundary-flux}
\end{equation}

The second-order field $V_{2,T}$ likewise contributes before the boundary
limits are taken.  The coefficient of $\kappa^2$ in
Eq.~\eqref{eq:app-area-second-order} is
\begin{align}
 4G_N\mathcal K^{(\kappa)}_{ij,\mathrm{nc}}
 ={}&-\frac12\int_\Omega a\,
 \nabla W_{ij}\mathbin{\cdot}\nabla N
 \notag\\
 &+\frac12\int_\Omega a\left(
 \nabla F_i\mathbin{\cdot}\nabla V_j
 +\nabla F_j\mathbin{\cdot}\nabla V_i\right)
 +\frac14\int_\Omega ag\,
 \nabla F_i\mathbin{\cdot}\nabla F_j,
 \label{eq:app-pre-green-pair}
\end{align}
where $g=|\nabla N|^2$.  The common profile $F_0$ has canceled.  The three
terms come from the pair component of $U_{1,T}$, the second-order field
$V_{2,T}$, and the direct quadratic coupling between $F_i$ and $F_j$.

The $W_{ij}$ contribution is already in flux form in
Eq.~\eqref{eq:app-W-boundary-flux}.  Since
$\mathcal DF_i=\mathcal DF_j=0$, the $V_2$ contribution is also a
Dirichlet-to-Neumann boundary term, namely a boundary integral pairing field
values with their outward normal derivatives.  Applying Green's identity
twice to the
remaining integral gives
\begin{equation}
 2\int_\Omega ag\,
 \nabla F_i\mathbin{\cdot}\nabla F_j
 =\int_\Omega H_{ij}\,\mathcal Dg
 +\int_{\partial\Omega}a\left(
 g\partial_nH_{ij}-H_{ij}\partial_ng\right).
 \label{eq:app-double-green}
\end{equation}
Define $R_i:=V_i-NF_i$ and similarly for $R_j$.  Substituting these
rearrangements into Eq.~\eqref{eq:app-pre-green-pair} gives the exact
finite-cutoff identity
\begin{align}
 4G_N\mathcal K^{(\kappa)}_{ij,\mathrm{nc}}
 ={}&\frac18\int_\Omega H_{ij}\,\mathcal Dg
 +\frac18\int_{\partial\Omega}a\left[
 (g+4N)\partial_nH_{ij}-H_{ij}\partial_ng\right]
 \notag\\
 &+\frac12\int_{\partial\Omega}a\left(
 R_j\partial_nF_i+R_i\partial_nF_j\right)
 \notag\\
 &-\frac12\int_{\partial\Omega}aW_{ij}\partial_nN.
 \label{eq:app-full-boundary-remainder}
\end{align}
This identity is proved before taking any boundary limit.  Its first term is
the bulk contribution in Eq.~\eqref{eq:K-positive-representation}; every
other term is an explicit integral over the regulated boundary.

The boundary integrals must be estimated rather than set to zero.
Equation~\eqref{eq:app-higher-jacobi-data} gives
$V_i=NF_i$, and therefore $R_i|_{z=0}=0$, at the conformal boundary; the same
holds for $j$.  Likewise,
Eq.~\eqref{eq:app-W-zero-data} gives $W_{ij}|_{z=0}=0$.  These statements
hold at $z=0$, not at a finite UV surface, and they say nothing by themselves
about the remote boundary.  The corresponding fluxes must remain in
Eq.~\eqref{eq:app-full-boundary-remainder} until their limits are estimated.
Thus the higher embedding fields are necessary: omitting them would miss
both the $4N\partial_nH_{ij}$ term and the $W_{ij}$ flux.

\subsection{Regulator limits}

To reduce \eqref{eq:app-full-boundary-remainder} to the bulk representation,
we impose the following UV asymptotics, an additional estimate for $W_{ij}$,
and an integrated decay condition at the remote boundary.  The last
condition does not follow from the NEC.

For the UV analysis, take
\begin{align}
 a(z)&=\frac{L_{\rm UV}^2}{z^2}
 \bigl(1+O(z^\nu)\bigr),
 \notag\\
 F_\ell(y,z)&=f_\ell(y)+\frac{z^2}{2}f_\ell''(y)
 +\mathcal L_\ell(y,z)+z^3\phi_\ell(y)+o(z^3).
 \label{eq:app-UV-expansions}
\end{align}
Here the Dirichlet jet means the boundary value together with its
derivatives.  The term $\mathcal L_\ell$ depends locally on these data and
vanishes wherever the data vanish.  ``Polyhomogeneous'' allows powers of
$z$ multiplied by powers of $\log z$, while a boundary collar is a fixed
small-$z$ neighborhood.  Since the two boundary supports are separated,
their local boundary data do not overlap.  Uniformly in such neighborhoods,
one therefore has
\begin{equation}
 H_{ij}=F_iF_j=O(z^3),
 \qquad \partial_zH_{ij}=O(z^2).
 \label{eq:app-H-UV}
\end{equation}

The zero Dirichlet data of $R_\ell$ make its polyhomogeneous expansion begin
with local $O(z^2)$ terms supported on the boundary data of $f_\ell$; the
first nonlocal term is $O(z^3)$.  Separation of the supports controls the two
cross terms involving $R_i$ and $R_j$.

The pair field $W_{ij}$ in Eq.~\eqref{eq:app-W-zero-data} needs a separate
assumption.  We require $W_{ij}=O(z^3)$ uniformly in the separated boundary
collars, allow the same polyhomogeneous logarithms, and assume an integrable
bound in $y$.  Together with $\partial_zN=p=O(z)$, these conditions give
\begin{align}
 aR_j\partial_zF_i&=O(z^2),
 \qquad
 aR_i\partial_zF_j=O(z^2),
 \notag\\
 aW_{ij}\partial_zN&=O(z^2),
 \label{eq:app-R-UV-flux}
\end{align}
up to logarithms when exponents in the near-boundary recursion coincide.
Hence all three fluxes from the
second-order embedding fields vanish as the UV regulator is removed.  Zero
Dirichlet data alone would not suffice.

It remains to control the two boundary terms containing $H_{ij}$.  With
$n(0)=0$, Eq.~\eqref{eq:p-regular} gives, for some integer $r\geq0$,
\begin{equation}
 p=-2z+O\!\left(z^{1+\bar\nu}(1+|\log z|)^r\right),
 \qquad
 n=-z^2+O\!\left(z^{2+\bar\nu}(1+|\log z|)^r\right),
 \qquad \bar\nu:=\min\{\nu,1\}.
\end{equation}
The logarithm allows for such resonances between exponents; when no resonance
occurs, one may take $r=0$.  The terms $p^2$ and $4n$ then cancel at leading
order:
\begin{equation}
 g+4N=p^2+4n
 =O\!\left(z^{2+\bar\nu}(1+|\log z|)^r\right).
 \label{eq:app-g-plus-4N}
\end{equation}
Together with $\partial_zg=O(z)$, this estimate implies
\begin{equation}
 a(g+4N)\partial_nH_{ij}=o(1),
 \qquad
 aH_{ij}\partial_ng=o(1),
\end{equation}
pointwise with an integrable uniform bound.  Hence the UV boundary integral
in Eq.~\eqref{eq:app-full-boundary-remainder} vanishes.  The local
counterterm cross terms vanish separately because the two boundary jets have
disjoint support.  In three boundary dimensions there is no smooth-curve
logarithmic anomaly.

At the remote boundary we assume the integrated decay
\begin{align}
 \lim_{\ell\to\infty}\int_{\partial\Omega_\ell^{\rm rem}}a\bigl(&
 |(g+4N)\partial_nH_{ij}|+|H_{ij}\partial_ng|
 \notag\\[-1mm]
 &+|R_j\partial_nF_i|+|R_i\partial_nF_j|
 +|W_{ij}\partial_nN|\bigr)=0.
 \label{eq:app-remote-flux-assumption}
\end{align}
The NEC does not imply this condition.  Together with the UV assumptions, it
makes every boundary term in
Eq.~\eqref{eq:app-full-boundary-remainder} vanish as the regulators are
removed, leaving Eq.~\eqref{eq:K-positive-representation}.

\subsection{Compact surface and positivity}
\label{app:compact-sct-limit}

The noncompact derivation differentiates at fixed remote cutoff and removes
that cutoff afterward.  On the compact SCT surface, the cone closes at
$\ell\sim|\kappa|^{-1}$ for nonzero $\kappa$, so the remote scale depends on
the deformation parameter and the required limit is diagonal.

To compare the two limits, let $A_\ell$ be the central region retained by the
cutoff, $C_\ell$ the closing cap, and $B_i,B_j$ the two bump regions.  The cap
changes the truncated pair combination by
\begin{equation}
 \Theta_{ij}(\ell;\kappa)
 :=I(C_\ell:B_iB_j\mid A_\ell)
  -I(C_\ell:B_i\mid A_\ell)
  -I(C_\ell:B_j\mid A_\ell).
 \label{eq:Theta-cap-short}
\end{equation}
Relating this correction to the local response requires two assumptions:
the fixed-cutoff expansion is uniform enough to exchange its $\kappa^2$
coefficient with the diagonal limit $\ell=|\kappa|^{-1}$, and the lower-order
cap terms cancel.  The entropy identity then gives
\begin{equation}
 \mathcal K^{(\kappa)}_{ij,\mathrm{compact}}
 =\mathcal K^{(\kappa)}_{ij,\mathrm{nc}}
 +\lim_{\kappa\to0}
 \frac{\Theta_{ij}(|\kappa|^{-1};\kappa)}{\kappa^2}.
 \label{eq:compact-local-relation}
\end{equation}
The strict noncompact result would transfer to the compact surface if, in
addition,
\begin{equation}
 \Theta_{ij}(|\kappa|^{-1};\kappa)=o(\kappa^2).
 \label{eq:compact-cap-estimate}
\end{equation}
Equation~\eqref{eq:compact-cap-estimate} is not proved, and the NEC alone
does not control the diagonal limit.  The compact extension therefore
remains open, while Theorem~\ref{thm:local-pair-strictness} applies only to
the noncompact branch.

The sign and degree bounds used in the main text follow directly from the
field equations.  Differentiating $(ap)'=2a$ gives
\begin{equation}
 p'=2+2\frac{h'}h p,
 \end{equation}
and therefore
\begin{align}
 \mathcal Dg
 &=8a+2(app')'
 \notag\\
 &=4a\left[
 \left(2+\frac{h'}h p\right)^2
 +\frac{h''}h p^2\right].
\end{align}
The first term is a square and the second is nonnegative by the NEC, giving
the nonnegative density used in
Theorem~\ref{thm:local-pair-strictness}.  For exact AdS,
$h=z/L$, $p=-2z$, $n=-z^2$, and $g=-4N$, so the density vanishes.  Once the
boundary limits have been taken, the response vanishes as well.

Boolean degree counts powers of the indicators that record which bumps
belong to $T$.  For finitely many disjoint bumps, $F_T$ has degree one and
$N$ has degree zero.  The equation for $V_{2,T}$ is linear in $F_T$, while
the equation and boundary value for $U_{1,T}$ are at most quadratic in
$F_T$ and $f_T$.  Consequently, each
structure in
Eq.~\eqref{eq:app-area-second-order} -- $U_{1,T}\cdot N$,
$F_T\cdot V_{2,T}$, and $|\nabla F_T|^2g$ -- has degree at most two.
The local counterterms and the operations at the remote boundary do not
increase the degree.  It follows that $\Xi_{\mathrm{nc}}(T)$ obeys
Eq.~\eqref{eq:sct-degree-two}.  This is a statement only about the
order-$\kappa^2$ response, not about the full entropy.

\clearpage
\phantomsection

\end{document}